\documentclass[a4paper]{article}

\usepackage{a4wide}
\usepackage{latexsym}
\usepackage[UKenglish]{babel}
\usepackage{graphicx,subfigure}
\usepackage{algorithm}
\usepackage{algpseudocode}
\usepackage{epstopdf}
	\graphicspath{{./}{figures/}}
\usepackage[colorlinks=true,linkcolor=blue,citecolor=blue]{hyperref}
\usepackage{xcolor}
\usepackage{amsfonts,amsmath,amsthm,mathtools,bbm,cool}

\newtheorem{theorem}{Theorem}[section]
\newtheorem{proposition}[theorem]{Proposition}
\newtheorem{corollary}{Corollary}[theorem]

\theoremstyle{remark}\newtheorem{remark}[theorem]{Remark}

\newcommand{\be}{\begin{equation}}
\newcommand{\ee}{\end{equation}}
\newcommand{\C}{\mathcal{C}}

\newcommand{\pd}{\partial}
\newcommand{\II}{\mathcal{I}}
\newcommand{\CC}{\mathcal{C}}
\newcommand{\rev}{\textcolor{black}}

\newcommand\restr[2]{{
		\left.\kern-\nulldelimiterspace 
		#1 
		\littletaller 
		\right|_{#2} 
}}

\newcommand{\littletaller}{\mathchoice{\vphantom{\big|}}{}{}{}}

\usepackage{hhline}
\usepackage{tabularx} 
\usepackage{multirow}

\usepackage{makecell}

\usepackage{todonotes}

\allowdisplaybreaks

\begin{document}
\title{
Effects of heterogeneous opinion interactions in many-agent systems for epidemic dynamics}

\author{
	Sabrina Bonandin \\
	{\small	Institute of Applied Mathematics (IGPM)},
		{\small RWTH Aachen University, Germany} \\
		{\small\tt bonandin@eddy.rwth-aachen.de} \\
	Mattia Zanella \\
		{\small	Department of Mathematics ``F. Casorati''},
		{\small University of Pavia, Italy} \\
		{\small\tt mattia.zanella@unipv.it} 
		}  
\date{\today}

\maketitle

\begin{abstract}
In this work we define a kinetic model for understanding the impact of heterogeneous opinion formation dynamics on epidemics. The considered many-agent system is characterized by nonsymmetric interactions which define a coupled system of kinetic equations for the evolution of the opinion density in each compartment. In the quasi-invariant limit we may show positivity and uniqueness of the solution of the problem together with its convergence towards an equilibrium distribution exhibiting bimodal shape. The tendency of the system towards opinion clusters is further analyzed by means of numerical methods, which confirm the consistency of the kinetic model with its moment system whose evolution is approximated in several regimes of parameters.  
\medskip


 \end{abstract}

\tableofcontents

\section{Introduction}
\label{sec:intro}

The mathematical modelling for the spread of infectious diseases \rev{traces} back to the pioneering works of D. Bernoulli and \rev{has} been \rev{made increasingly} more sophisticated over the centuries. Amongst the most influential approaches to mathematical epidemiology, the Kermack-McKendrick SIR model dates back to the first half of the 20th century \cite{KMcK}. In general terms, compartmental modelling relies on the subdivision of the population into epidemiologically relevant groups, where each group represents a stage of progression \rev{in the individual's health} with respect to the transmission dynamics  \cite{Het00}. More recently, several extensions of the SIR-type model have been proposed to incorporate behavioural aspects into these model, see \cite{BDM,BDMDOG,FSJ} and the references therein. However, a complete understanding of the multiscale features of epidemic dynamics should take into account the heterogeneous scales driving the infection dynamics. In this direction, kinetic equations for collective phenomena are capable to link the microscopic scale of agents with the macroscopic scale of  observable data. In particular, suitable transition rates have been determined in relation to emerging social dynamics \cite{DPertTZ}, together with the definition of possible control strategies \cite{DTZ}. 

The study of kinetic models for large interacting systems has gained increasing interest during the last decades \cite{ABB+22,CFTV,HT,PT13}. Amongst the most studied emerging patterns in many-agent systems, aggregation dynamics gained increased interest thanks to its the widespread applications in heterogeneous fields, see \cite{CFL,CG,DMLM,HK,PPDT,SWS}. In particular, a solid theoretical framework suitable for investigating emerging properties of opinion formation phenomena by means of mathematical models has been provided by classical kinetic theory  since the formation of a relative consensus is determined by elementary variations of the agents' opinion converging to an equilibrium distribution under suitable assumptions \cite{CT,DMPW,DW,PTTZ,Tos06,TTZ}.

During the recent pandemic it has been observed how, as cases of infection escalated, the collective adherence to the so-called non-pharmaceutical interventions was crucial to ensure public health in the absence of effective treatments \cite{Gatto,Flocco_etal}. Recent experimental results have shown that social norm changes are often triggered by opinion alignment phenomena \cite{Tu}. In particular, the perceived adherence of individuals’ social network has a strong impact on the effective support of the protective behaviour. Hence, the individual responses to \rev{threats} are a core question to set up effective measures prescribing norm changes in daily social contacts and have deep connection with vaccination hesitancy. With the aim of understanding the impact of opinion formation in epidemic dynamics,  several models have been proposed to determine the evolution of the opinion of individuals on protective measures in a multi-agent system under the spread of an infectious disease \cite{BTZ,Flocco_etal,KGFPN,Zan23}. The study of opinion formation phenomena is also closely connected with the problem of vaccination hesitancy \cite{FPBGB} and the propagation of misinformation on the agents' contact network \cite{PB09}.

 In this work, we concentrate on a kinetic compartmental model to investigate the emergence of collective structures triggered by nonsymmetric interactions between agents in different compartments. In this direction we expand the results in \cite{Zan23} taking advantage of the kinetic epidemic setting developed in \cite{DPertTZ,DPTZ20}. \rev{Indeed, we will show how heterogeneous opinion exchanges in multi-agent systems may lead to the formation of opinion clusters even for unpolarised societies, having a direct impact on the evolution of the epidemic. Furthermore, we}  derive new macroscopic equations encapsulating the effects of opinion phenomena in epidemic dynamics at the level of observable quantities. In particular, we show that the emergence of opinion clusters in the form of bimodal Beta distributions can be ignited by the coupled action of opinion and epidemic dynamics. Furthermore, we \rev{provide} proofs of positivity and uniqueness of the solution for a surrogate Fokker-Planck-type model. 
 
In more \rev{detail}, the paper is organized as follows: in Section \ref{sec:coupledmodel}, we introduce the kinetic compartmental model 
and we derive \rev{its} constituent properties.
Hence, in Section \ref{sec:FP}, we derive a reduced complexity operator of Fokker-Planck type complemented with no-flux boundary conditions to understand the emerging opinion patterns from the many-agent system and we prove the positivity and the uniqueness of the solution to the corresponding Fokker-Planck system. In Section \ref{sec:macro}, we derive a macroscopic system of equations and we exploit the new model to prove that the kinetic epidemic system possesses an explicitly computable steady state. In Section \ref{sec:num}, we perform several numerical experiments based on a recently developed structure preserving scheme. 

\section{The kinetic model}
\label{sec:coupledmodel}

In this section, we introduce a kinetic compartmental model suitable to describe the evolution of opinion of individuals on protective measures in a multi-agent system under the spread of an infectious disease. 

We consider a system of agents that is subdivided into the following four epidemiologically relevant compartments: susceptible ($S$),  individuals that can contract the infection; infected ($I$), infected and infectious agents; exposed ($E$), infected agents that are not yet infectious, and recovered ($R$), agents that were in the compartment $I$ and that cannot contract the infection. We assume that the time scale of the epidemic dynamics is sufficiently rapid, so that demographic effects - such as entry or departure from the population - may be ignored: as a direct consequence, the total population size constant can be considered constant.
In addition, we equip each agent of the population with an opinion variable $w \in [-1,1] = \mathcal I$, where the boundaries of the interval $\mathcal I$ denote the two extreme opposite opinions. In particular, if an individual has opinion $w=-1$, it means that he/she does not believe in the adoption of protective measures (e.g., social distancing and masking), while, on the contrary, $w=1$ is linked to maximal approval of protection. Hereafter, we let $\II = [-1,1]$ be the interval of all admissible opinions. Last, we assume that agents with a high protective behavior are less likely to contract the infection and that the exchange of opinions on protective measures is influenced by the stage of progression in the individual's health.\\

We denote by $f_H = f_H(w,t): [-1,1]\times \mathbb R_+ \to \mathbb R_+$ the distribution of opinions at time $t\ge0$ of agents in the compartment $H \in \CC = \{ S,E,I,R \}$ such that $f_H(w,t)dw$ represents the fraction of agents with opinion in $[w,w+dw]$ at time $t\ge0$ in the compartment $H$. Without \rev{loss} of generality thanks to the conservation of the total population size, we impose that 
\begin{equation}
\label{eq:condition_mass}
\sum_{H \in \CC} f_H(w,t) = f(w,t), \quad \int_\II f(w,t)dw = 1.
\end{equation}

For each time $t\ge0$, we define the mass fraction of agents in $H\in \mathcal C$ and their moment of order $r>0$ to be the quantities
\[
\rho_H(t) = \int_\II f_H(w,t)dw, \quad  \rho_H(t)m_{r,H}(t) = \int_\II w^r f_H(w,t)dw, 
\]
respectively.  In the following, to simplify notations, we will use $m_{1,H}(t) = m_H(t)$ for the (local) mean opinion at time $t\ge 0$ in class $H$. 

The kinetic compartmental model characterising the coupled evolution of opinions and infection is given by the following system of kinetic equations
\begin{equation}
	\label{eq:kin_cs}
	\begin{dcases}
		\pd_t f_S(w,t) = -K(f_S,f_I)(w,t) + \frac{1}{\tau} \sum_{J \in \CC}Q_{SJ}(f_S,f_J)(w,t)\\
		\pd_t f_E(w,t) = K(f_S,f_I)(w,t) -\nu_E f_E +\frac{1}{\tau}\sum_{J \in \CC}Q_{EJ}(f_E,f_J)(w,t) \\
		\pd_t f_I(w,t) = \nu_E f_E(w,t) - \nu_I f_I +\frac{1}{\tau}\sum_{J \in \CC}Q_{IJ}(f_I,f_J)(w,t) \\
		\pd_t f_R(w,t) = \nu_I f_I(w,t) +\frac{1}{\tau}\sum_{J \in \CC}Q_{RJ}(f_R,f_J)(w,t) \\
	\end{dcases}
\end{equation}
for any $w \in \II$ and $t \ge 0$. 
Having a close look at the system, we immediately recognize the presence of two distinct time scales, the scale of epidemiological dynamics and the one characterising opinion formation phenomena. The parameter denotes $\tau > 0$ the frequency at which the agents modify their opinion in response to the epidemic dynamics. In \eqref{eq:kin_cs} we introduced the operators $Q_{HJ}(\cdot, \cdot)$ characterising the thermalization of the distribution $f_H$ towards its local equilibrium distribution in view of the interaction dynamics with agents of compartments $J \in \CC$. Furthermore, in \eqref{eq:kin_cs} the parameter $\nu_E>0$ is determined by the latency and $\nu_I>0$ is  the recovery rate, see e.g. \cite{BCCF19}. Finally,  the operator $K(\cdot,\cdot)$ is the local incidence rate, which is given by
\begin{equation}
	\label{def:K}
	K(f_S,f_I)(w,t) = f_S(w,t) \int_{\II} \kappa(w,w_*) f_I(w_*,t)dw_*
\end{equation}
for any $w \in \II, t \ge 0$ and with $\kappa(\cdot,\cdot)$ being the contact rate between people of opinion $w$ and $w_*$. Several choices can be made to model $\kappa(\cdot,\cdot)$, in the following we will consider the following
\begin{equation}
	\label{def:kappa}
	\kappa(w,w_*) = \frac{\beta}{4^{\alpha}} (1-w)^{\alpha}(1-w_*)^{\alpha},
\end{equation}
where $\beta > 0$ is a baseline transmission rate characterizing the epidemics and $\alpha \ge 0$ is a coefficient linked to the efficacy of the protective measures. The choice in \eqref{def:kappa} synthesize the assumption that agents with opinion close to $-1$, i.e. to non protective behaviour, are more likely to contract the infection. These dynamics have been \rev{proposed} in literature for vaccine hesitancy, see e.g. \cite{PB09}. 


We may observe that if $\alpha=0$ the transition between compartments is given by the simplified  operator
\begin{equation}
	\label{def:K,alpha=0}
	K(f_S,f_I)(w,t) = \beta f_S(w,t) \rho_I(t),
\end{equation}
in which we do not observe any effect of opinion formation dynamics on the epidemic dynamics. Indeed, a direct integration of \eqref{eq:kin_cs} with respect to the opinion variable gives the classical SEIR model for the system of masses $\rho_J(t)$, $J \in \mathcal C$. 
On the other hand,  the case in which $\alpha=1$ leads to a local incidence rate of the form
\begin{equation}
	\label{def:K,alpha=1}
	K(f_S,f_I)(w,t) = \frac{\beta}{4}(1-w)f_S(w,t)(1-m_I(t))\rho_I(t).
\end{equation}
which highlights the dependence of the transition between epidemiological compartments on the behaviour of infectious agents and in particular on their mean opinion.

\subsection{A kinetic model for opinion formation dynamics}
\label{subs:opinion_model}

Coherently with the modeling approach of \cite{Zan23}, we let the opinion dynamics in kinetic compartmental system \eqref{eq:kin_cs} be described by the kinetic model of continuous opinion formation introduced in \cite{Tos06}. 
The model is based on binary interactions (hence, the mathematical methods we use are close to those used in the context of kinetic theory of granular gases \cite{CIP13}) and assumes that the formation of opinion is made up of two distinct processes: the compromise process, that reflects the human tendency to settle conflicts, and the diffusion process, that comprises all the unpredictable opinion deviations that an agent might have in response to global access to information. \\

We recall that the \rev{novelty} of the model we are proposing (compared to the one of \cite{Zan23}) is that exchange of opinion on protective measures occurs between agents of any compartment. Hence, we consider now two agents, one belonging to compartment $H$, endowed with opinion $w$, and one to compartment $J$, endowed with opinion $w_*$.
The post-interaction opinion pair $(w', w'_*) \in \II^2$ of two interacting agents is given by
\begin{equation}
	\label{eq:binary_rules}
	\begin{cases}
		w' = w + P(w,w_*) (\gamma_J w_* - \gamma_H w)
		+ \eta_H D(|w|) \\
		w'_* = w_*  + P(w,w_*)(\gamma_H w - \gamma_J w_*) w+ \eta_{J} D(|w_*|)
	\end{cases}
\end{equation}
where $P(\cdot,\cdot)$ is the interaction function and $P(\cdot,\cdot)\in[0,1]$, $\gamma_H, \gamma_J \in (0,1)$ are compartment-dependent compromise propensities,  and $\eta_H, \eta_J$ are iid random variables such that $\left\langle \eta_H \right\rangle = \left\langle \eta_J \right\rangle =0$ and variance  $\left\langle \eta_H^2 \right\rangle = \left\langle \eta_J^2 \right\rangle =\sigma^2 >0$. At last, 
the local relevance of the diffusion is given by  $D(w)\ge0$. 
We have 
\begin{equation}
\label{eq:average}
\begin{split}
\left\langle w^\prime - w \right\rangle  = -P(w,w_*)(\gamma_H w - \gamma_J w_*) \\
\left\langle w_*^\prime - w_* \right\rangle  = - P(w,w_*)(\gamma_Jw_* - \gamma_H w)
\end{split}
\end{equation}
 therefore, if $\gamma_J w_* > \gamma_H w$, from the first equation in \eqref{eq:average} we have $\left\langle w^\prime - w \right\rangle >0$ implying in average that $ \left\langle w^\prime  \right\rangle > w$. At the same time, from the second equation we get $\left\langle w_*^\prime - w_* \right\rangle <0$, implying in average that $ \left\langle w_*^\prime  \right\rangle < w_*$. 
We remark that the assumptions made on $\gamma_H, \gamma_J, \eta_H, \eta_J, P, D$ are not sufficient to guarantee that $w', w'_* \in \II$ unless $\eta_{H}\equiv \eta_J\equiv 0$. A sufficient condition to guarantee that $(w^\prime,w^\prime_*)\in \II^2$ is that two constants $c_H,c_J>0$ exist such that 
\begin{equation*}
	|\eta_{H}| \le c(1-\gamma_H), \qquad |\eta_{J}| \le c(1-\gamma_J),  \end{equation*}
and 
\[
c_HD(|w|) \le 1-|w|, \qquad  c_JD(|w|) \le 1-|w|, 
\]
for any $w \in [-1,1]$. We point the interested reader to \cite{Tos06,Zan23} for a detailed proof. 

Assuming the introduced bounds on the random variables in \eqref{eq:binary_rules} we may determine the evolution of the distribution $f_H(w,t)$, $H \in \mathcal C$, through the methods of kinetic theory for many-agent systems \cite{CIP13}. In particular, the evolution of the kinetic density is obtained by means of a Boltzmann-type equation 
\begin{equation}
	\label{eq:opinion_kin_eq}
	\pd_t f_H (w,t)= \sum_{J \in \CC } Q_{HJ}(f_H,f_J)(w,t)
\end{equation}
\rev{with}
\[
\begin{split}
&Q_{HJ}(f_H,f_J)(w,t)\\ 
&\quad=\left\langle \int_{\II} \left(  \frac{1}{{}^\prime\mathcal J} f_H({}^\prime w,t)f_J({}^\prime w_*,t) -  f_H(w,t) f_J(w_*,t) \right) dw_* \right\rangle,
\end{split}
\]
where $({}^\prime w,{}^\prime w_*)$ are pre-interaction opinions generating the post-interaction opinions $(w,w_*)$ and ${}^\prime\mathcal{J}$ is the determinant of the Jacobian of the transformation $({}^\prime w,{}^\prime w_*)\to (w,w_*)$. 

\begin{remark}

\rev{In the microscopic interactions of \cite{Tos06} the terms related to the compromise propensity} are both governed by the same constant $\gamma$. This, in particular, \rev{gives to $\gamma$ the interpretation of a} shared compromise propensity \rev{between the two agents exchanging their opinion}. 
In the compartmental extension of \rev{such} opinion formation model, we \rev{keep this hypothesis by assuming} that each agent of a compartment shares the same compromise propensity. 
\end{remark}


\subsection{Evolution of observable quantities}
\label{subs:prop_opinion}

In the previous section, we introduced the microscopic model for opinion formation and the
corresponding kinetic equation. In order to derive the surrogate Fokker-Planck equation in Section \ref{sec:FP}, in this subsection, we look at what macroscopic quantities are conserved in time by the model.

Let $\phi(w), w \in \II,$ denote a test function.
The weak formulation of  kinetic equation \eqref{eq:opinion_kin_eq} is given for each $H \in \mathcal C$ by  
\begin{align}
	\frac{d}{dt} \int_{\II} \phi (w) f_H(w,t) dw
	=& \sum_{J \in \CC } \int_\II \phi(w) Q_{HJ}(f_H,f_J)(w,t) dw \nonumber \\
	\label{eq:opinion_weak}
	=& \sum_{J \in \CC} \Bigl< \int_{\II^2} \left[ \phi(w') - \phi(w)  \right] f_H(w,t) f_J(w_*,t)  dw dw_* \Bigr>,
\end{align}
where $\left\langle \cdot \right\rangle $ denotes the expected value with respect to the distribution of the random variable.
Choosing $\phi(w) = 1,w,w^2$, we are able to infer the evolution of observable quantities like \rev{the total number of agents, their mean opinion within each compartment and the second order moment}. 

If $\phi(w) = 1$ we get  the conservation of mass. If $\phi(w) = w$ from \eqref{eq:opinion_weak} we get
\[
\dfrac{d}{dt} (\rho_H m_H) = \sum_{J \in \mathcal C} \int_{\mathcal{I}^2} P(w,w_*)(\gamma_J w_* - \gamma_H w)f_H(w,t)f_J(w_*,t)dw\, dw_*, 
\] 
and the mean opinion is not conserved in time. 
In the simplified case $P\equiv 1$ we get
\[
\dfrac{d}{dt} (\rho_H m_H(t)) = \sum_{J \in \mathcal C} \rho_H \rho_J (\gamma_J m_J(t) - \gamma_H m_H(t)) = \rho_H \left( M(t) - \gamma_H  m_H(t) \right),
\] 
where
\begin{equation}
	\label{def:M}
	M(t) := \sum_{J \in \CC } \gamma_J \rho_J m_J(t) 
\end{equation}
is the total weighted mean opinion at time $t\ge0$. Hence, 
the total mean opinion, that is definied as the sum over the compartments of the local mean opinions, is a conserved quantity since we get
\begin{align*}
	\frac{d}{dt} \left( \sum_{H\in \mathcal C} \rho_H m_H \right) 
	&= \sum_{H\in \mathcal C}  \left( \frac{d}{dt} \rho_H m_H \right) = \sum_{H\in \mathcal C}  \rho_H \left( M - \gamma_H  m_H \right)\\
	&= \sum_{H\in \mathcal C} \rho_H M - \sum_{H\in \mathcal C} \gamma_H \rho_H m_H = \sum_{H\in \mathcal C} \rho_H M - M = 0
\end{align*}
in view of \eqref{eq:condition_mass}. Therefore, \rev{unlike in} \cite{Zan23}, the mean opinion is not conserved for symmetric interaction functions. 
If $\phi(w) = w^2$, the evolution of the energy of is given by 
\begin{align*}
	&\frac{d}{dt}\int_\II w^2 f_H(w,t) dw
	= \sum_{J \in \CC } \int_{\mathcal I^2} \left\langle (w')^2-w^2 \right\rangle f_H(w,t) f_J(w_*,t)  dw dw_*  \\
	&= \sum_{J \in \CC} \left[ \int_{\II^2} (\gamma_H^2 P^2(w,w_*) - 2\gamma_H P(w,w_*)) w^2 f_H(w,t) f_J(w_*,t) dw dw_* \right] \\ 
	&+ \sum_{J \in \CC }  \left[ \gamma^2_J  \int_{\II^2} P^2(w,w_*) w^2_* f_H(w,t) f_J(w_*,t) dw dw_* + \sigma^2 \rho_J \int_\II D^2(|w|) f_H(w,t)dw \right] \\
	&+ \sum_{J \in \CC }  \left[ 2 \gamma_J \int_{\II^2} (1-\gamma_H P(w,w_*)) P(w,w_*) w w_* f_H(w,t) f_J(w_*,t) dw dw_* \right].
\end{align*}
As in \cite{Zan23}, we conclude that energy is not conserved by the model.	
In the case of $P \equiv 1$, it can be equivalently written as
\begin{align*}
	\frac{d}{dt} (\rho_H m_{2,H}(t)) & = (\gamma^2_H-2\gamma_H) \rho_H m_{2,H}(t) + \rho_H \sum_{J \in \CC } \gamma_J^2  \rho_Jm_{2,J}(t) \\
	&+ \sigma^2 \int_\II D^2(|w|)  f_H(w,t)dw + 2(1-\gamma_H) \rho_H(t) m_H(t) M(t),
\end{align*}
and, again, we see that it is not conserved.


\section{Derivation of a reduced complexity Fokker-Planck model}
\label{sec:FP}

A \rev{closed-form} analytical derivation of the equilibrium distribution of the Boltzmann-type collision operator $Q_{HJ}$ in \eqref{eq:opinion_kin_eq} is difficult. 
For this reason, several reduced complexity models have been proposed. In this section, we consider a scaling of compromise and diffusion that has its roots in the so-called grazing collision limit of the classical Boltzmann equation \cite{CIP13, PT13, Tos06}. In the following, we will assume that $\phi$ is 3-H\"older continuous and $\eta_H, \eta_J$ have finite third order moments, see \cite{Tos06}.\\

Let $\epsilon > 0$ be a scaling coefficient. We introduce the following scaling 
\begin{equation}
	\gamma_J \to \epsilon \gamma_J, \quad \sigma^2_J \to \epsilon \sigma^2_J
	\label{eq:qi_regime}
\end{equation}
and, in the time scale $\xi=\epsilon t$, we introduce the corresponding scaled distributions
\[
g_H(w,\xi) = f_H(w,t) = f_H\left(w, \xi/\epsilon \right),\quad H \in \mathcal C.
\]
In the following we will indicate with $\rho_H m_H(\xi) = \int_{\mathcal I} w g_H(w,\xi)dw$. Rewriting the weak formulation \eqref{eq:opinion_weak} of the opinion kinetic equation \eqref{eq:opinion_kin_eq} for the scaled function, we get 
\begin{equation}
	\label{eq:opinion_weak,eps}
	\epsilon \frac{d}{d\xi} \int_{\II} \phi (w) g_H(w,\xi) dw 
	=  \sum_{J \in \CC }  \int_{\II^2} \left\langle  \phi(w') - \phi(w)  \right\rangle  g_H(w,\xi) g_J(w_*,\xi)  dw dw_*.
\end{equation}
Letting $\epsilon \to 0^+$,  we can introduce a Taylor expansion of $\phi$ around $w$ 
\[
\langle \phi(w') - \phi(w) \rangle = \langle w'-w \rangle\dfrac{d}{dw} \phi(w) + \frac{1}{2} \langle (w'-w)^2 \rangle \dfrac{d^2}{dw^2} \phi(w) + \frac{1}{6} \langle (w'-w)^3 \rangle \dfrac{d^3}{dw^3} \phi(\bar{w})
\]
where $\bar{w} \in (\min(w,w'), \max(w,w'))$ and
\begin{align*}
	\langle w'-w \rangle &= \epsilon P(w,w_*) (- \gamma_H w  +  \gamma_J w_*) \\
	\langle (w'-w)^2 \rangle &= \epsilon^2 P^2(w,w_*) (\gamma_H^2 w^2  +  \gamma_J^2 w^2_* - 2  \gamma_H \gamma_J w w_*) + \epsilon \sigma^2 D^2(|w|).
\end{align*}
Plugging these terms in \eqref{eq:opinion_weak,eps} we get 
\begin{equation}
\label{eq:partialFP}
\begin{split}
&\epsilon\dfrac{d}{d\xi} \int_{\mathcal I}\phi(w)g_H(w,\xi)dw \\
& = \epsilon \sum_{J \in \mathcal C} \int_{\II^2} \phi'(w) P(w,w_*) (-\gamma_H w+ \gamma_J w_*) g_H(w,\xi)g_J(w_*,\xi) dw dw_*\\
&\quad+ \frac{1}{2} \sum_{J \in \mathcal C} \int_{\II^2} \phi''(w) \langle (w'-w)^2 \rangle g_H(w,\xi)g_J(w_*,\xi) dw dw_* \\
&\quad + \sum_{J \in \mathcal C}R_{\epsilon}(g_H,g_J)(\xi),
\end{split}
\end{equation}
and 
\[
R_{\epsilon}(g_H,g_J)(\xi) = \frac{1}{6} \int_{\mathcal I^2} \langle (w'-w)^3 \rangle  \phi^{\prime\prime\prime}(\bar{w}) g_H(w) g_J(w_*) dw dw_*. 
\]
%
Hence, we may observe that for each $J \in \mathcal C$, the reminder term is such that 
\[
\dfrac{1}{\epsilon}|R_\epsilon(g_H,g_J)| \to 0,
\]
for $\epsilon\to 0^+$ since $\left\langle\eta_J^3\right\rangle<+\infty$ for all $J \in \mathcal C$. Therefore, in the quasi-invariant scaling, 
letting $\epsilon \to 0^+$ in \eqref{eq:partialFP}, we get 
\begin{equation}
	\label{eq:FP_opinion,weak}
	\begin{split}
		&\frac{d}{d\xi} \int_{\II} \phi(w)g_H(w,\xi) dw \\
		&= \int_\II \left\lbrace  \phi^\prime \int_\II P(w,w_*) \sum_{J \in \CC} \left( \gamma_J w_* - \gamma_H w \right) g_J(w_*) dw_* + \phi^{\prime\prime} \frac{\sigma^2}{2} D^2(w) \right\rbrace  g_H(w) dw  
	\end{split}
\end{equation}
Integrating back by parts, in view of the smoothness of $\phi$, we obtain the surrogate Fokker-Planck operator 
\begin{equation}
	\label{eq:FP_opinion} 
	\partial_\xi g_H(w,\xi) = \bar Q_{H}(g_H)(w,\xi)
\end{equation}
where
\[
\begin{split}
&\bar Q_H(g_H)(w,\xi) = \frac{\sigma^2}{2} \pd^2_w \left(D^2(|w|)g_H(w, \xi)\right) \\
&\quad + \pd_w \left( \left(\int_\II P(w,w_*) \sum_{J \in \CC} \left( \gamma_H w - \gamma_J w_* \right) g_J(w_*, \xi) dw_*  \right) g_H(w, \xi) \right)
\end{split}
\]
complemented with the no-flux boundary conditions
\begin{equation}
	\label{eq:BC_FP_opinion}
	\begin{split} 
		\restr{\left(\int_\II P(w,w_*) \sum_{J \in \CC} \left( \gamma_H w - \gamma_J w_* \right) g_J(w_*, \xi) dw_*  \right) g_H(w, \xi)}{w=\pm 1} \\
		+ \restr{\frac{\sigma^2}{2} \pd_w ((D^2(|w|)g_H(w, \xi))}{w=\pm 1} = 0 \\
		\restr{D^2(|w|)g_H(w, \xi)}{w=\pm 1} = 0 
	\end{split}
\end{equation}
\noindent for any $\xi \ge 0$. We observe that these conditions express a balance
between the so-called advective and diffusive fluxes on the boundaries $w = \pm 1$. 

\begin{remark}
In the simplified case in which $P \equiv 1$, using \eqref{eq:condition_mass}, it is straightforward to deduce that Fokker-Planck equation \eqref{eq:FP_opinion} can be rewritten as
\begin{equation}
	\label{eq:FP_opinion,P=1} 
	\partial_\xi g_H(w,\xi) = \frac{\sigma^2}{2} \pd^2_w \left(D^2(|w|)g_H(w, \xi)\right) + \pd_w ((\gamma_H w - M(\xi))g_H(w,\xi)
\end{equation}
where  $M(\xi)$ has been defined in \eqref{def:M}. Therefore, under the additional assumption $\gamma_H = \gamma$ and  $D(|w|) = \sqrt{1-w^2}$, we can compute the explicit steady state of \eqref{eq:FP_opinion,P=1}. Indeed, \eqref{eq:FP_opinion,P=1} simplifies into the following Fokker-Planck-type model
\begin{equation*}
	\frac{\pd g_H(w,\xi)}{\pd \xi} = \frac{\sigma^2}{2} \pd^2_w \left((1-w^2))g_H(w, \xi)\right) + \pd_w \left( (\gamma w-M(\xi)) g_H(w, \xi) \right),
\end{equation*}
where now $M(\xi) = \bar{M} = \gamma \sum_{J \in \mathcal C} \rho_J m_J$ which is a conserved quantity. For large times, we get
\[
g_H^\infty(w) = \rho_H^\infty \dfrac{(1-w)^{-1 + \frac{1-\bar M}{\lambda}}(1+w)^{-1 + \frac{-1+\bar M}{\lambda}}}{\rev{\textrm{B}((1-\bar M)/\lambda, (1+\bar M)/\lambda)}},
\]
where \rev{$\lambda = \sigma^2/\gamma$}. 
\end{remark}

\subsection{Properties of the model}
\label{subs:analyt_FP}

We consider the kinetic compartmental model with opinion formation term given by the derived Fokker-Planck model \eqref{eq:FP_opinion} and where the local incidence rate is given either by \eqref{def:K,alpha=0} or by \eqref{def:K,alpha=1}. Without loss of generality, in the following, we restore the time variable $t\ge0$. We get
\begin{equation}
	\begin{cases}
		\pd_{t} g_S(w,t) = -K(g_S,g_I)(w,t) +\dfrac{1}{\tau} \bar{Q}_S(g_S)(w,t) \\
		\pd_{t} g_E(w,t) = K(g_S,g_I)(w,t) -\nu_E g_E(w,t) + \dfrac{1}{\tau}\bar{Q}_E(g_E)(w,t) \\
		\pd_{t} g_I(w,t) = \nu_E g_E(w,t) - \nu_I g_I(w,t) + \dfrac{1}{\tau}\bar{Q}_I(g_I)(w,t) \\
		\pd_{t} g_R(w,t) = \nu_I g_I(w,t) +\dfrac{1}{\tau} \bar{Q}_R(g_R)(w,t). 
	\end{cases}
	\label{ivp:FP_system}
\end{equation}
In this subsection we \rev{first prove the positivity of the solution} to \eqref{ivp:FP_system} with no-flux boundary conditions \eqref{eq:BC_FP_opinion}, given positive $g_H(w,0) =  g^0_H \in L^1(\II)$, for all $H \in \mathcal C$. \rev{Then, under the same hypotheses on the initial data, but with the additional assumption of constant interaction forces $ P \equiv 1$, we prove the uniqueness of such solution.}

\paragraph{Positivity of the solution to \eqref{ivp:FP_system}.}
In order to prove the positivity of the solution, we adopt a time-splitting strategy by isolating the opinion dynamics and the epidemiological one. Hence, the first problem is obtained by 
\begin{equation}
	\begin{cases}
		\pd_t g_H(w,t) = \bar{Q}_H(g_H)(w,t) \\
		\rev{g_H(w, 0) = g^0_H(w) \quad  \:H \in \CC},\\ 
		\textnormal{No-flux boundary conditions \eqref{eq:BC_FP_opinion}},
	\end{cases}
	\label{ivp:opinion_kin_eq}
\end{equation}
for all $H \in \mathcal C$, while the second one by 
\begin{equation}
	\label{ivp:epid}
	\begin{cases}
		\pd_{t} g_S(w,t) = -K(g_S,g_I)(w,t) \\
		\pd_{t} g_E(w,t) = K(g_S,g_I)(w,t) \\
		\pd_{t} g_I(w,t) = \nu_E g_E(w,t) - \nu_I g_I(w,t) \\
		\pd_{t} g_R(w,t) = \nu_I g_I(w,t) \\ 
		g_H(w, 0) = g^0_H(w) \quad  \:H \in \CC. 
	\end{cases}
\end{equation}

We begin by proving the positivity of the solution to \eqref{ivp:opinion_kin_eq}. We exploit the arguments of \cite{CRS08} and \cite{FMZ23} and derive it as a corollary of the theorem that follows.
\begin{proposition}[\textbf{Non-increase of the $L^1$ norm}]
	\label{th:non-incr_L1norm}
	Let $g_H(w,t)$ be a solution of \eqref{ivp:opinion_kin_eq}. If $g^0_H \in L^1(\II)$, then $\int_{\II} |g_H(w,t)|dw = \|g_H(\cdot, t)\|_{L^1(\II)}$ is non increasing for any $t \ge 0$.
\end{proposition}
\begin{proof}
	Let $\epsilon >0$. We denote by $\text{sign}_{\epsilon}(g_H)$ a regularized increasing approximation of the sign function (e.g., a sigmoid, such as the hyperbolic tangent) and by $|g_H|_{\epsilon}$ the regularized approximation of $|g_H|$ via the primitive of $\text{sign}_{\epsilon}(g_H)$.\\
	Given weak formulation \eqref{eq:FP_opinion,weak}, we introduce for $w \in \II$
	\begin{equation*}
		A(w, t) = \sum_{J \in \CC} \int_\II P(w,w_*)  \left( \gamma_J w_* - \gamma_H w \right) g_J(w_*) dw_*, \quad B(w) = \frac{\lambda}{2}D^2(|w|).
	\end{equation*}
	Hence, we obtain
	\begin{equation*}
		\frac{d}{dt} \int_{\II} \phi(w) g_H(w,t) dw = \int_{\II} \left[ A(w, t) \phi^\prime(w) + B(w) \phi^{\prime\prime}(w) \right] g_H(w,t) dw.
	\end{equation*}
	\rev{If} we choose $\phi(w) = \text{sign}_{\epsilon}(g_H)$ in the above equation and \rev{avoid} the dependence on $w\in\mathcal I$ and $ t\ge0$, we \rev{obtain} 
	\begin{equation*}
		\frac{d}{dt} \int_{\II} \text{sign}_{\epsilon}(g_H) g_H \; dw = \int_{\II} \left[ A \;  \pd_w (\text{sign}_{\epsilon}(g_H)) +  B \; \pd_w^2 (\text{sign}_{\epsilon}(g_H)) \right]  g_H \; dw.
	\end{equation*}
We have
\[
	\begin{split}
		&\frac{d}{dt} \int_{\II} |g_H |_{\epsilon} \; dw =  \\
		&\quad \int_{\II} A g_H  \; \text{sign}'_{\epsilon}(g_H ) \; \pd_w g_H  \; dw + \int_\II B g_H  \; \pd_w[ \text{sign}'_{\epsilon}(g_H ) \; \pd_w g_H ] \; dw =  \\
		&\quad \int_{\II} A g_H  \; \text{sign}'_{\epsilon}(g_H ) \; \pd_w g_H  \; dw + \restr{\left[ B g_H  \; \text{sign}'_{\epsilon}(g_H) \; \pd_w g_H \right]}{w = \pm 1} \\
		&\quad - \int_\II \pd_w[Bf] \; \text{sign}'_{\epsilon}(g_H) \; \pd_w g_H \; dw  = \\
		&\quad \int_{\II} A g_H \; \text{sign}'_{\epsilon}(g_H) \; \pd_w g_H \; dw - \int_\II \pd_wB \; g_H \; \text{sign}'_{\epsilon}(g_H) \; \pd_w g_H \; dw \\
		&\quad- \int_\II B \; g_H \; \text{sign}'_{\epsilon}(g_H) \; (\pd_w g_H)^2 \; dw
	\end{split}
	\]
	where we integrated by parts the second addend of the first equation and we used that 
	$\restr{\left[ B g_H \; \text{sign}'_{\epsilon}(g_H) \; \pd_w g_H \right]}{w = \pm 1}$ is vanishing in view of the second no-flux boundary condition in \eqref{eq:BC_FP_opinion}. Observing that $ \pd_w [ g_H \; \text{sign}_{\epsilon}(g_H) - |g_H|_{\epsilon}] = g_H \; \text{sign}'_{\epsilon}(g_H) \; \pd_w g_H $, the weak formulation finally reads
	\begin{equation*}
		\begin{split}
			\frac{d}{dt} \int_{\II} |g_H|_{\epsilon} \; dw = \int_{\II} (A-\pd_w B) \pd_w [ g_H \; \text{sign}_{\epsilon}(g_H) - |g_H|_{\epsilon}] \; dw \\
		- \int_\II B \; g_H \; \text{sign}'_{\epsilon}(g_H) \; (\pd_w g_H)^2 \; dw.
		\end{split}
	\end{equation*}
	Integrating by parts the first addend of the right-hand side and using the first no-flux boundary conditions in \eqref{eq:BC_FP_opinion}, we have that
	\begin{equation*}
		\begin{split}
			\frac{d}{d\xi} \int_{\II} |g_H|_{\epsilon} \; dw = \int_{\II} \pd_w(A-\pd_w B) \;[ g_H \; \text{sign}_{\epsilon}(g_H) - |g_H|_{\epsilon}] \; dw \\
			- \int_\II B \; g_H \; \text{sign}'_{\epsilon}(g_H) \; (\pd_w g_H)^2 \; dw.
		\end{split}
	\end{equation*}
	Therefore, in the limit $\epsilon \to 0^+$ we obtain 
	\begin{equation*}
		\frac{d}{dt} \int_{\II} |g_H(w,t)| \; dw = \frac{d}{dt} \|g_H(\cdot, t)\|_{L^1(\II)} \le 0.
	\end{equation*}
	
\end{proof}

\begin{corollary}[\textbf{Positivity of the solution to \eqref{ivp:opinion_kin_eq}}]
	\label{th:pos_FP_opinion}
	Let $g_H$ be a solution of \eqref{ivp:opinion_kin_eq}. If $g^0_H \in L^1(\II)$ and $g^0_H(w) \ge 0$, then $g_H(w,\xi) \ge 0$ for any $w \in \II, \xi \ge 0$.
\end{corollary}
\begin{proof}
		The result follows from the proof presented in \cite{CRS08} and \rev{from Proposition \ref{th:non-incr_L1norm}.}
		\end{proof}

Now can prove the positivity of the solution to \eqref{ivp:epid} by distinguishing the scenarios in which $\alpha=0$ and $\alpha=1$ (and, thus, when $K(f_S,f_I)$ is of form \eqref{def:K,alpha=0} and \eqref{def:K,alpha=1} respectively). 

\begin{proposition}[\textbf{Positivity of the solution to \eqref{ivp:epid}}]
	\label{th:pos_epid}
	Let $g_H$, $H \in \CC$ be a solution of the initial-value problem \eqref{ivp:epid}. 
	If $g^0_H(w) \ge 0$, then $g_H(w,t) \ge 0$ for any $w \in \II, t \ge 0$.
\end{proposition}
\begin{proof}
		The result follows from the proof presented in \cite{BTZ,FMZ23}.
		\end{proof}

Merging the positivity results in \rev{Corollary} \ref{th:pos_FP_opinion} and Proposition \ref{th:pos_epid} we can provide positivity of the solution to \eqref{ivp:FP_system}.
\begin{proposition}[\textbf{Positivity of the solution to \eqref{ivp:FP_system}}]
	Let $g_H$, $H \in \CC$ be a solution of \eqref{ivp:FP_system}. If $g^0_H \in L^1(\II)$ and $g^0_H(w) \ge 0$, then $g_H(w,\xi) \ge 0$ for any $ w \in \II, \xi \ge 0$.
\end{proposition}
\begin{proof}
We can discretize equation \eqref{ivp:FP_system}  through a classical splitting method \cite{Temam} in time. We briefly recall the splitting strategy. For any given time $T>0$ and $n\in \mathbb N$, we introduce a time discretization $t^k = k\Delta t$, $k \in [0,n]$ with $\Delta t = T/n>0$. Then we proceed by solving two separate problems in each time step as follows:
\begin{itemize}
\item At time $t=0$ we consider $g_H(w,0) = g_H^0(w)\ge 0$, $g_H^0\in H^1(\mathbb R)$, for all $H \in \mathcal C$. 
\item For $t \in [t^k,t^{k+1}]$ we solve the Fokker-Planck step 
\[
\begin{split}
\partial_t g_H(w,t) &= \bar Q_H(g_H)(w,t),\\
g_H(w,t^k) &= g_H^k(w) 
\end{split}\]
for all $H \in \mathcal C$.
\item The solution of the Fokker-Planck step at time $t^{k+1}$ is assumed as the initial value for the epidemiological step in the same time interval $t \in [t^k,t^{k+1}]$. 
\item For $t\in [t^k,t^{k+1}]$ the epidemiological step is subsequently solved by considering \eqref{ivp:epid}. 
\end{itemize}
The method generates an approximation $g_{H,n}(w,t)$ of the solution to \eqref{ivp:FP_system}, for which properties can be easily derived by resorting to the properties of the Fokker-Planck and epidemiological steps, which are solved in sequence. \rev{Hence, we may proceed as in \cite{APZ} making use of the Trotter's formula, which allows to conclude that 
\[
\lim_{n\to+\infty}g_H^n(w) = g_H(w,t)\ge0
\] 
and this shows the positivity of $g_H(w,t)$. The approach is reminiscent of  the one developed in \cite{DM}.}
\end{proof}

\paragraph{Uniqueness of the solution to \eqref{ivp:FP_system}.}
\rev{In this subsection we additionally require $P \equiv 1$. Then, Fokker-Planck equation \eqref{eq:FP_opinion} reduces to \eqref{eq:FP_opinion,P=1} and the operator on the right-hand side becomes linear in $g_H$. }
 We remark that the contact rate $\kappa(w,w_*)$ as in \eqref{def:kappa} is bounded. Indeed, 
\rev{$0 \le \kappa(w,w_*) \le\beta$ for any $w,w_*\in [-1,1]$ and any $\alpha\ge 0$}.
We get the following result 
\begin{theorem}[\textbf{Uniqueness of the solution to \eqref{ivp:FP_system}}]
	Let $g_H, \bar{g}_H$, $H \in \CC$ be two solutions of \eqref{ivp:FP_system} \rev{with $P \equiv 1$}. If $g^0_H, \bar{g}^0_H \in L^1(\II)$, then there exists $C^{max} = C^{max}(\beta,\nu_E,\nu_I) > 0$ such that for any $t \ge 0$
	\begin{equation*}
		\sum_{J \in \CC} ||g_H(\cdot,t) - \bar{g}_H(\cdot, t) ||_{L^1(\II)} \le e^{C^{max} t} \sum_{H \in \CC} ||g_H^0 - \bar{g}_H^0 ||_{L^1(\II)}
	\end{equation*}
\end{theorem}
\begin{proof}
	The result follows from the proof presented in \cite{BTZ,FMZ23}. The proof is based
	\rev{on the fact that $g_H -\bar{g}_H$ is a solution to \eqref{ivp:FP_system}, thanks to the linearity of the Fokker-Planck operator in $g_H$, and that consequently Proposition \ref{th:non-incr_L1norm} may be applied to $g_H -\bar{g}_H$.}
	We remark that, at variance with the the just-mentioned \rev{papers} where the boundness of the contact rate was imposed by the authors, here $\kappa$ is bounded by definition, as shown in the calculations preceding the theorem.
\end{proof}

\section{Evolution of the moment system for the opinion-based SEIR model}
\label{sec:macro}

As remarked in Section \ref{sec:FP}, 
the drift term in surrogate Fokker-Planck equation \eqref{eq:FP_opinion} depends on time. 
This makes the mathematical analysis of the corresponding four-equation system in \eqref{ivp:FP_system} more challenging.
As we're interested in drawing conclusions on the macroscopic epidemic trends resulting from the model, in this section, we derive the system for the evolution of the mass fractions and local mean opinions and explain how these can be used to prove that \eqref{ivp:FP_system} possesses an explicitly computable steady state. 
From now on, we restrict to the scenario of constant interaction forces $P \equiv 1$, so that in particular the total mean opinion of the model is conserved as proven in Subsection \ref{subs:prop_opinion}.\\

Let us consider first the case $\alpha=0$. Then, $\kappa(w,w_*) \equiv \beta$ and the local incidence rate $K(f_S,f_I)(w,t)$ is of form \eqref{def:K,alpha=0}. Kinetic compartmental system \eqref{eq:kin_cs} reduces to
\begin{equation}
	\label{eq:kin_cs,alpha=0}
	\begin{dcases}
		\pd_t f_S = -\beta f_S \rho_I +\dfrac{1}{\tau}  \sum_{J \in \{ S,E,I,R \} }  Q_{SJ}(f_S,f_J) \\
		\pd_t f_E = \beta f_S \rho_I -\nu_E f_E + \dfrac{1}{\tau}  \sum_{J \in \{ S,E,I,R \} }  Q_{EJ}(f_E,f_J) \\
		\pd_t f_I = \nu_E f_E - \nu_I f_I + \dfrac{1}{\tau}  \sum_{J \in \{ S,E,I,R \} }  Q_{IJ}(f_I,f_J) \\
		\pd_t f_R = \nu_I f_I + \dfrac{1}{\tau} \sum_{J \in \{ S,E,I,R \} }  Q_{RJ}(f_R,f_J). \\
	\end{dcases}
\end{equation}
Integrating system \eqref{eq:kin_cs,alpha=0} with respect to $w\in \mathcal I$ we get
\begin{equation}
	\label{eq:SEIR}
	\begin{dcases}
		\frac{d}{dt} \rho_S(t) = - \beta \rho_S(t) \rho_I(t)  \\
		\frac{d}{dt} \rho_E(t) = \beta \rho_S(t) \rho_I(t) - \nu_E \rho_E(t)\\
		\frac{d}{dt} \rho_I(t) = \nu_E \rho_E(t) - \nu_I \rho_I(t)\\
		\frac{d}{dt} \rho_R(t) = \nu_I \rho_I(t)\\
	\end{dcases}
\end{equation}
which is the classical $SEIR$ compartmental model. 
Multiplying system \eqref{eq:kin_cs,alpha=0} by $w$ and integrating with respect to the $w$ variable, we obtain the system for the evolution of the mean opinions 
\begin{equation}
	\label{eq:mean_op}
	\begin{dcases}
		\frac{d}{dt}(\rho_Sm_S) = - \beta\rho_I\rho_Sm_S +  \dfrac{\rho_S}{\tau}(M(t) - \gamma_S m_S(t)) \\
		\frac{d}{dt}(\rho_Em_E) = \beta \rho_I\rho_Sm_S - \nu_E\rho_Em_E +  \dfrac{\rho_E}{\tau} (M(t) - \gamma_E m_E(t)) \\
		\frac{d}{dt}(\rho_Im_I) = \nu_E \rho_E m_E - \nu_I \rho_I m_I + \dfrac{\rho_I}{\tau}  (M(t) - \gamma_I m_I(t)) \\
		\frac{d}{dt}(\rho_Rm_R) = \nu_I \rho_I m_I(t)  + \dfrac{\rho_R}{\tau} (M(t) - \gamma_R m_R(t))
	\end{dcases}
\end{equation}
where we recall that $M$ is given by \eqref{def:M}. 

On the other hand, if we let $\alpha = 1$, the local incidence rate $K(f_S,f_I)(w,t)$ is of the form \eqref{def:K,alpha=1} and the kinetic compartmental model \eqref{eq:kin_cs} has the following form 
\begin{equation}
	\label{eq:kin_cs,alpha=1}
	\begin{dcases}
		\pd_t f_S = -\dfrac{\beta}{4}(1-w) f_S(1-m_I) \rho_I +\dfrac{1}{\tau}  \sum_{J \in \{ S,E,I,R \} }  Q_{SJ}(f_S,f_J) \\
		\pd_t f_E =\dfrac{\beta}{4}(1-w) f_S(1-m_I) \rho_I -\nu_E f_E + \dfrac{1}{\tau}  \sum_{J \in \{ S,E,I,R \} }  Q_{EJ}(f_E,f_J) \\
		\pd_t f_I = \nu_E f_E - \nu_I f_I + \dfrac{1}{\tau}  \sum_{J \in \{ S,E,I,R \} }  Q_{IJ}(f_I,f_J) \\
		\pd_t f_R = \nu_I f_I + \dfrac{1}{\tau} \sum_{J \in \{ S,E,I,R \} }  Q_{RJ}(f_R,f_J). \\
	\end{dcases}
\end{equation}
Hence, integrating \eqref{eq:kin_cs,alpha=1} with respect to $w \in \mathcal I$ we get
\begin{equation}
\label{eq:mass,alpha=1}
\begin{cases}
\dfrac{d}{dt}\rho_S(t) &=  - \dfrac{\beta}{4}(1-m_I)(1-m_S)\rho_I \rho_S \\
\dfrac{d}{dt}\rho_E(t) &=\dfrac{\beta}{4}(1-m_I)(1-m_S)\rho_I \rho_S - \nu_E\rho_E \\
\dfrac{d}{dt}\rho_I(t) &= \nu_E \rho_E - \nu_I \rho_I \\
\dfrac{d}{dt}\rho_R(t) &= \nu_I \rho_I, 
\end{cases}
\end{equation}
whose evolution now depends on the first moment of the kinetic densities $f_S(w,t)$, $f_I(w,t)$. A direct inspection on the evolution of the moment system is obtained by multiplying \eqref{eq:kin_cs,alpha=1} by $w \in \mathcal I$ and integrating with respect to the opinion variable to get

\begin{equation}	\label{eq:mean_op,alpha=1}
\begin{cases}
\dfrac{d}{dt}{(\rho_S m_S)} &= - \dfrac{\beta}{4}\rho_I(1-m_I)\displaystyle \int_{\mathcal I} w(1-w)f_S(w,t)dw \\
&+ \dfrac{\rho_S}{\tau}(M(t) - \gamma_S m_S) \\
\dfrac{d}{dt}{(\rho_E m_E)} &=\dfrac{\beta}{4}\rho_I(1-m_I)\displaystyle \int_{\mathcal I} w(1-w)f_S(w,t)dw - \nu_E \rho_E m_E \\
&+ \dfrac{\rho_E}{\tau}(M(t) - \gamma_E m_E) \\
\dfrac{d}{dt}{(\rho_I m_I)}  &= \nu_E \rho_E m_E - \nu_I \rho_I m_I   + \dfrac{\rho_I}{\tau}(M(t) - \gamma_I m_I)\\
\dfrac{d}{dt}{(\rho_R m_R)} &= \nu_I \rho_I m_I + \dfrac{\rho_R}{\tau}(M(t) - \gamma_R m_R) ,
\end{cases}
\end{equation}
which depends on the kinetic density $f_S(w,t)$. Unlike what presented in \cite{Zan23} we cannot rely on a closure strategy since the mean \rev{opinions} are not conserved quantities. 
\rev{
In more details, we observe that Fokker-Planck equation \eqref{eq:FP_opinion,P=1} admits quasi-stationary equilibrium states and that they may be obtained by simply imposing $\pd_\xi g_H(w,\xi) = 0$. 
However, it would not be exact to close the systems with their moments, as, by doing so, we would be closing the system with respect to quantities which are not conserved in time. Indeed, we recall again that our model conserves the total mean opinion $\sum_{H \in \C} \rho_H m_H$, but not the mean opinions $\rho_H m_H$ in each compartment $H$.}

\subsection{Stationary solutions in an explicitly solvable case}
\label{subs:stationary_sol}

We consider the kinetic compartmental model \eqref{eq:kin_cs,alpha=0} where the thermalization operators are now of Fokker-Planck-type. We have
\begin{equation}
	\label{eq:FP_system,alpha=0}
	\begin{dcases}
		\pd_{t} g_S = -\beta g_S \rho_I +\dfrac{1}{\tau} \bar{Q}_S(g_S) \\
		\pd_{t} g_E = \beta g_S {\rho}_I -\nu_E g_E + \dfrac{1}{\tau} \bar{Q}_E(g_E) \\
		\pd_{t} g_I = \nu_E g_E - \nu_I g_I + \dfrac{1}{\tau}\bar{Q}_I(g_I) \\
		\pd_{t} g_R = \nu_I g_I +\dfrac{1}{\tau} \bar{Q}_R(g_R). \\
	\end{dcases}
\end{equation}
Since $\alpha = 0$ and the system for the evolution of the mass fractions corresponds to the classical $SEIR$ compartmental model, we can use standard results on the large time behaviour of the solution to such model (see for instance \cite{BCCF19, Het00}). In particular, \begin{equation}
	\label{eq:SEIR_limit}
	\lim_{t \to \infty} \rho_S(t) = \rho_S^{\infty} > 0, \;
	\lim_{t \to \infty} \rho_E(t) =  \; \lim_{t \to \infty} \rho_I(t) = 0, \;
	\lim_{t \to \infty} \rho_R(t) = \rho_R^{\infty} > 0
\end{equation}
where $\rho_S^{\infty} + \rho_R^{\infty} = 1$. 
Then, merging the fact that the mass fractions of the exposed and the infected vanish for large times with the evolution of the local mean opinions given by \eqref{eq:mean_op},  in the limit $t \to +\infty$,  we get
\begin{equation*}
\begin{split}
	\rho_S(t)m_S(t) \to \rho_S^\infty m_S^\infty, \\
	\rho_E(t)m_E(t)  \to 0, \\
	\rho_I(t)m_I(t) \to 0, \\
	\rho_R(t)m_R(t) \to \rho_R^\infty m_R^\infty
	\end{split}
\end{equation*}
with the asymptotic mean opinions $m_S^{\infty}, m_R^{\infty}$ satisfying
\begin{equation*}
	2M^{\infty} - \gamma_S m_S^{\infty} - \gamma_R m_R^{\infty} = 0
\end{equation*}
where $M^{\infty} = \gamma_S \rho_S^{\infty} m_S^{\infty} + \gamma_R \rho_R^{\infty} m_R^{\infty}$. Therefore, we have
\begin{equation}
	\label{eq:1st_constr}
	\gamma_S m_S^{\infty} = \gamma_R m_R^{\infty}. 
\end{equation}
Furthermore, we know from Subsection \ref{subs:prop_opinion} that the total mean opinion $m = \sum_J \rho_J m_J$ is conserved by the model. This, in particular, implies that
\begin{equation*}
	\rho_S^{\infty} m_S^{\infty} + \rho_R^{\infty} m_R^{\infty} = m. 
\end{equation*}
Hence, we are able to write $m_S^{\infty}, m_R^{\infty}$ as 
\begin{equation}
	\label{eq:asympt_mean_op}
	m_S^{\infty} = \frac{\gamma_R}{\gamma_R \rho_S^{\infty} + \gamma_S \rho_R^{\infty}} m, \quad
	m_R^{\infty} = \frac{\gamma_S}{\gamma_R \rho_S^{\infty} + \gamma_S \rho_R^{\infty}} m.
\end{equation}
We remark that, once the kinetic compartmental model is complemented with initial conditions, $m, \rho_S^{\infty}, \rho_R^{\infty}$ are quantities that are explicitly computable and, thanks to equation \eqref{eq:asympt_mean_op}, so are $m_S^{\infty}, m_R^{\infty}$.

Finally, for the Fokker-Planck operator with constant interaction $P\equiv 1$ \eqref{eq:FP_opinion,P=1} we get in the limit $\tau \to 0^+$ that the system reaches a  steady state distribution $g^\infty(w) = g^{\infty}_S(w) + g^{\infty}_R(w)$ where $g^\infty_H(w)$, $H \in \{S,R\}$ are determined for any $w \in \II$ as the solutions of the following system of differential equations
\begin{align*}
	& (\gamma_S w - \gamma_S m^{\infty}_S) g^{\infty}_S(w) + \frac{\sigma^2}{2} \pd_w [D^2(|w|)g^{\infty}_S(w)] = 0,  \\
	& (\gamma_R w - \gamma_R m^{\infty}_R) g^{\infty}_R(w) + \frac{\sigma^2}{2} \pd_w [D^2(|w|)g^{\infty}_R(w)] = 0.                                                                                                                                                                                                                                                                                                                                                                                                                                                                                                                                                                                                                                                                                                                                                                                                                                                                                                                                                                                                                                                                                                                                                                                                                                                                                                                                                                                                                                                                                                                                                                                                                                                                                                                                                                                                                                                                                                                                                                                                                                                                                                                                                                                                                                                                                                                                                                                                                                                                                                                                                                                                                                                                                                                                                                                                                                                                                                                                                                                                                        
\end{align*}
\rev{A proof for the existence of such steady state is provided in \cite{DPTZ20} and is based on the Fourier metrics introduced in \cite{BCT,GTW}.
Proceeding} as in \cite{Zan23}, in the relevant case $D(w) = \sqrt{1-w^2}$, the distributions $g^{\infty}_S(w)$ and $g^{\infty}_R(w)$ are explicitly computable \rev{and are of form}
\begin{equation}
	\label{eq:g_inf_S,R}
	g^{\infty}_H(w) = \rho_H^\infty \dfrac{ (1-w)^{-1 + \frac{1-m^{\infty}_H}{\lambda_H}} (1+w)^{-1 + \frac{1+m^{\infty}_H}{\lambda_H}}}{\textrm{B}((1-m_H^\infty)/\lambda_H,(1+m_H^\infty)/\lambda_H)},
\end{equation}
where $\textrm{B}(\cdot,\cdot)$ is the Beta function, $m_H^\infty$ is defined in \eqref{eq:asympt_mean_op} and where we indicated with $\lambda_H = \sigma^2/\gamma_H$,  $H \in\{ S,R\}$. \rev{For a review on other choices} of the diffusion function $D(|w|)$ we refer to \cite{Tos06}. We may observe that $g_H^\infty(w)/\rho_H$ defined in \eqref{eq:g_inf_S,R} is a Beta probability density. Furthermore, we may observe that the global steady state distribution $g^\infty(w)$ may exhibit a bimodal shape. 

As argued in \cite{Tos06} a Beta distribution has a peak in $\II$ when $\lambda = \sigma^2/\gamma < 1 - |m|$ and in correspondence to the point 
\[ \bar w = \frac{m}{1-\lambda}.
\]
Therefore, we expect to observe a bimodal shape for $g^{\infty}$ if both $\lambda_S<1-|m^\infty_S|$ and $\lambda_R< 1-|m_R^\infty|$ or, equivalently, if $\sigma^2/\gamma_S<1-|m^\infty_S|$ and $\sigma^2/\gamma_R< 1-|m_R^\infty|$. 
In addition, we recall that $\gamma_S, \gamma_R, m^{\infty}_S, m^{\infty}_R$ are linked by relation \eqref{eq:1st_constr}.
All in all, the five parameters $\sigma^2, \gamma_S, \gamma_R, m^{\infty}_S, m^{\infty}_R$ shall satisfy
\begin{equation}
\label{eq:2max_condition}
\begin{dcases}
	\frac{\sigma^2}{\gamma_S} < 1 - |m_S^{\infty}|\\
	\frac{\sigma^2}{\gamma_R} < 1 - |m_R^{\infty}|\\
	\gamma_S m_S^{\infty} = \gamma_R m_R^{\infty} \quad \text{with $m_S^{\infty}m_R^{\infty}> 0$}
\end{dcases}
\end{equation}
where the constraint on the product $m_S^{\infty}m_R^{\infty}$ comes from the fact that $\sigma^2, \gamma_S, \gamma_R >0$ by  their definition. In the top row of Figure \ref{fig:bimodal_cases} we give two sets of parameters that satisfy the above conditions and for which we see a bimodal shape.
It is interesting to observe that multi-modal distributions are obtained through Beta densities, at variance with \cite{DPTZ20} where multi-modal distributions were obtained through Gamma ones.

\begin{figure}
	\centering
	\subfigure[]{
		\includegraphics[scale = 0.48]{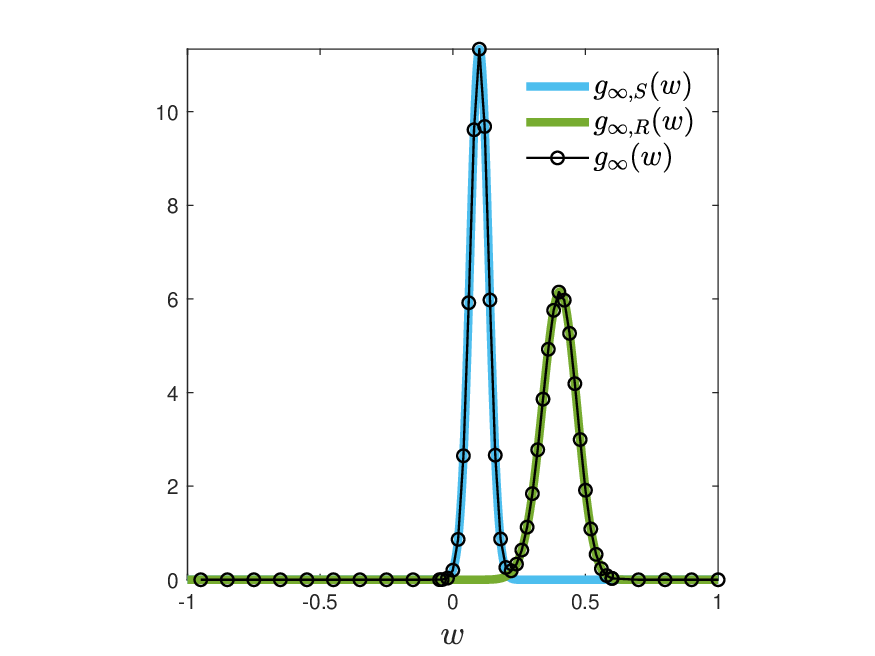}}
	\subfigure[]{
		\includegraphics[scale = 0.48]{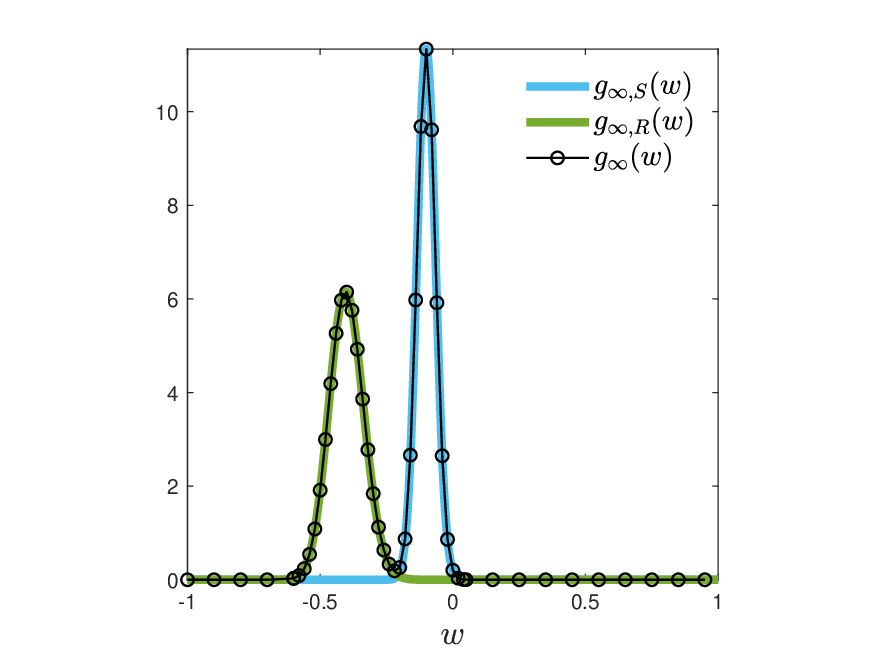}}	\\
	\subfigure[]{
		\includegraphics[scale = 0.48]{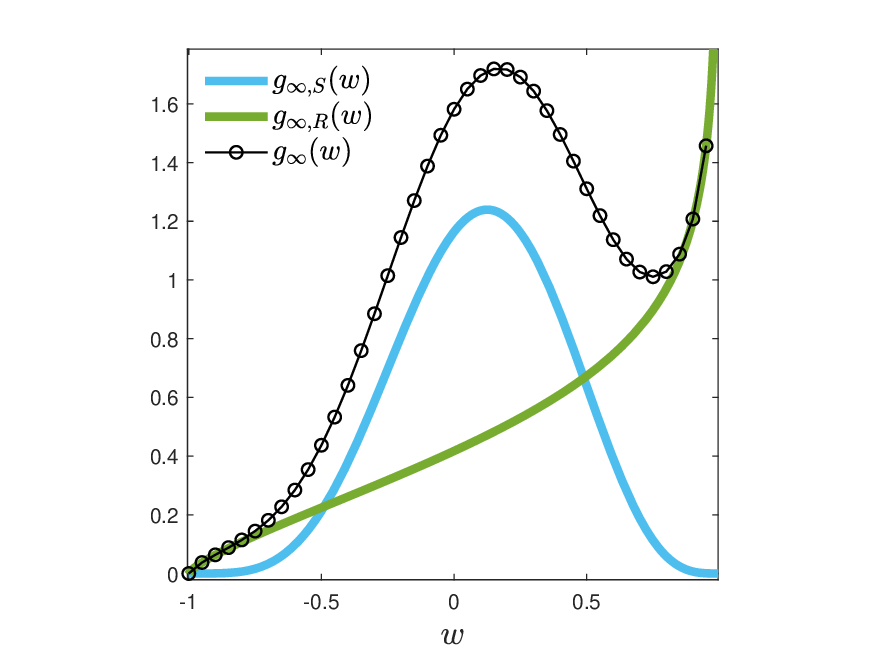}}
	\subfigure[]{
		\includegraphics[scale = 0.48]{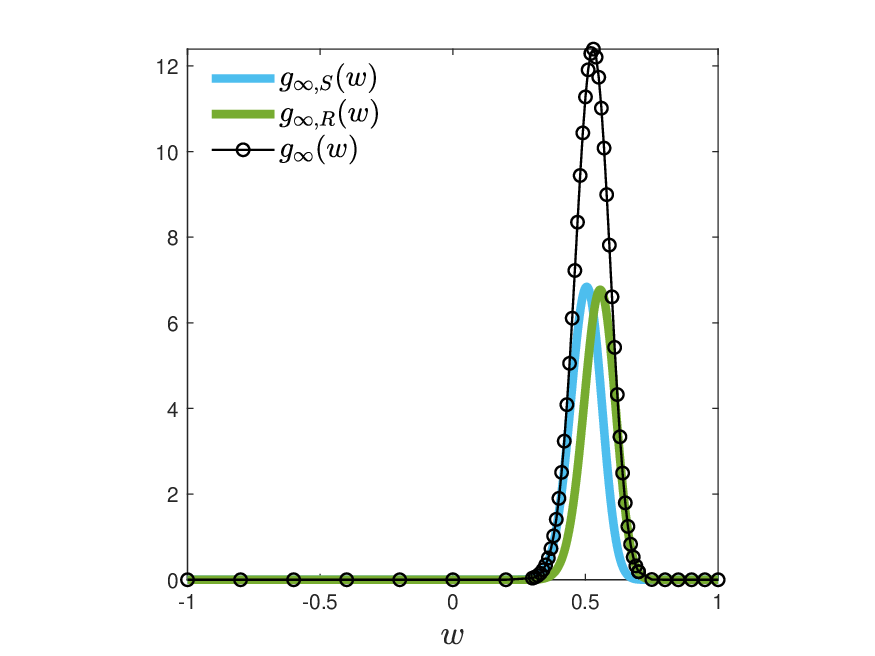}}
	\caption{
		Plot of the global steady state $g^{\infty}$ for various choices of the opinion and epidemiological parameters. In all the plots we fix $\sigma^2 = 10^{-3}$.
		The plot on the top-left corner $(a)$  is obtained by choosing $\gamma_S=0.8,\gamma_R=0.2,m^{\infty}_S=0.1, m^{\infty}_R=0.4$, so that consensus-type dynamics for $S$ and $R$ is observed and \eqref{eq:2max_condition} is verified: as expected $g^{\infty}$ presents two maxima in $\II$. 
		The plot on the top-right corner $(b)$ is obtained with the same choices of compromise-propensity parameters as before and in the case $m^\infty_S = -0.1$, $m_R^\infty = -0.4$. 
		The plot on the bottom-left corner $(c)$  is obtained by choosing the same asymptotic mean opinions as in the plot above it, but with $\gamma_R= 0.0025$ such that the constraint $\sigma^2/\gamma_R < 1 - |m^{\infty}_R|$ is not satisfied (that is, so that $g^{\infty}_R$ exhibits opinion polarization of a society). $\gamma_S$ is then calculated using \eqref{eq:1st_constr}.	
		The plot on the bottom-right corner $(d)$  is obtained by choosing $\gamma_S=0.1,\gamma_R=0.1,m^{\infty}_S=0.5, m^{\infty}_R=0.55$. 
		In this scenario we obtain a uni-modal steady profile and conclude that \eqref{eq:2max_condition} are not a sufficient condition for the existence of two peaks.}
	\label{fig:bimodal_cases}
\end{figure}

Clearly, if either $g^{\infty}_S$ or $g^{\infty}_R$ reveal opinion polarization of a society, then the global steady state has only one maximum in the interval $\II$, as shown, for instance, in the bottom-left corner of Figure \ref{fig:bimodal_cases}.
Finally, a question that arises spontaneous at this point is whether the existence of a maximum for both $g^{\infty}_S$ and $g^{\infty}_R$ implies a bimodal shape for $g^{\infty}$. The answer is negative and a counterexample is presented in the bottom-right corner of Figure \ref{fig:bimodal_cases}.

\begin{remark}
	The Fokker-Planck-type system \eqref{eq:FP_system,alpha=0} we obtained is capable \rev{of exhibiting} the formation of asymptotic opinion clusters even in the case of constant interactions. In opinion-formation phenomena, possible ways to observe the emergence of clusters is typically based on the adoption of bounded-confidence-type interactions functions, see \cite{HK} and \cite{BL,PTTZ} together with the references therein.
%
\end{remark}

\begin{remark}
	In this section, we restricted ourselves to the scenario in which $\alpha =0$. Indeed, as remarked in the first part of the section, this simplified assumption allows us to obtain a $SEIR$ model for the evolution of the local mass fractions and, thus, to use the classical results on the behaviour of its solution for large times. However, 
	we remark that an open question regards the formation of opinion clusters for $\alpha>0$. 
\end{remark}

\newpage
\section{Numerical results}
\label{sec:num}

In this section, we numerically test the consistency of the proposed modelling approach. Furthermore, we will investigate the impact of opinion segregation features on epidemic dynamics. 
From a methodological point of view, to approximate the kinetic SEIR model with Fokker-Planck-type operators, we resort to structure-preserving schemes for nonlinear Fokker-Planck equations \cite{PZ}. These methods are capable of preserving the main physical properties of the equilibrium density, like positivity, entropy dissipation and preservation of observable quantities. 

In more detail, we are interested in the evolution of $f_J(w,t)$, $J \in \mathcal C$, $w \in [-1,1]$, $t\ge0$ solution to \eqref{ivp:FP_system} and complemented by the initial conditions $f_J(w,0) = f_J^0(w)$. We consider a time discretization of the interval $[0,t_{\textrm{max}}]$ of size $\Delta t>0$. We will denote by $f_J^n(w)$ the approximation of $f_J(w,t^n)$. Hence, we may introduce a splitting strategy between the collision step $f_J^* = \mathcal O_{\Delta t}(f_J^n)$
\[
\begin{split}
\partial_t f_J^* &= \dfrac{1}{\tau}\bar Q_J(f_J^*),\\
f_J^*(w,0)  &= f_J^n(w), \qquad J \in \mathcal C,
\end{split}
\]
and the epidemiological step $f_J^{**} = \mathcal E_{\Delta t}(f_J^{**})$
\[
\begin{split}
\partial_t f_S^{**} &= -K(f_S^{**},f_I^{**}) \\
\partial_t f_E^{**} &= K(f_S^{**},f_I^{**}) - \nu_E f_E^{**} \\
\partial_t f_I^{**} &= \nu_E f_E^{**} - \nu_I f_I^{**}\\
\partial_t f_R^{**} &= \nu_I f_I^{**},\\
f_J^{**}(w,0) &= f_J^*(w,\Delta t), \qquad J \in \mathcal C. 
\end{split}
\] 
The operators $\bar Q_J(\cdot)$ have been defined in \eqref{eq:FP_opinion} and are complemented by no-flux conditions \eqref{eq:BC_FP_opinion}. We highlight that, at time $t^{n+1}$,  the solution is given by the combination of the two introduced steps. In the following we will adopts a second-order Strang splitting method that is obtained as
\[
f_J^{n+1}  = \mathcal E_{\Delta t/2}(\mathcal O_{\Delta t}(\mathcal E_{\Delta t/2}(f_J^n(w)))), 
\]
for all $J\in \mathcal C$. As introduced above, the Fokker-Planck step is solved by a semi-implicit SP method, whereas the integration of the epidemiological step is performed with an RK4 method. In the following, we will always assume $\tau = 1$.

\subsection{Test 1. Consistency between the kinetic model and the moment system}
\label{test1}
In this test we focus on the case $\alpha=0$ in \eqref{eq:kin_cs,alpha=0} such that 
\[
K(f_S,f_I)(w,t) = \beta f_S(w,t)\rho_I(t),
\]
and we compare the evolution of the derived moment system  \eqref{eq:SEIR}-\eqref{eq:mean_op} derived with constant interaction function $P\equiv 1$.  To define the initial conditions, we introduce the distributions
\begin{equation}
\label{eq:gh}
g_0(w) = 
\begin{cases}
0 & w \in [0,1], \\
1 & \textrm{elsewhere},
\end{cases}
\qquad 
h_0(w) = 
\begin{cases}
1 & w \in [0,1], \\
0 & \textrm{elsewhere}. 
\end{cases}
\end{equation}
In the following, we will consider the initial distributions
\begin{equation}
\label{eq:ic_test1a}
\begin{split}
f_S(w,0) = \rho_S(0) g_0(w), \qquad f_E(w,0) = \rho_E(0)g_0(w), \\
f_I(w,0) = \rho_I(0) h_0(w), \qquad f_R(w,0) = \rho_R(0)h_0(w),
\end{split}
\end{equation}
with $\rho_E(0)=\rho_I(0) =\rho_R(0) = 0.05$  and $\rho_S(0) = 1-\rho_E(0)-\rho_I(0) - \rho_R(0)$. The introduced initial conditions describe a society where the subsceptible agents have negative initial opinions on protective behaviour.  We solve numerically \eqref{eq:kin_cs,alpha=0} over the time frame $[0,t_{\textrm{max}}]$ by introducing a time discretization $t^n = n\Delta t$, $\Delta t>0$, and $n = 0,\dots,T$ such that $T\Delta t = t_{\textrm{max}}$. We further introduce a grid $w_i\in [-1,1]$ with $w_{i+1}-w_i=\Delta w>0$, $i=1,\dots,N_w$. In Figure \ref{fig:1} we report the evolution of the approximated kinetic densities where we further considered the epidemiological parameters $\beta = 0.3$, $\nu_E=1/2$, $\nu_I = 1/12$, whereas the compromise propensities are given by $\gamma_S = \gamma_E = 0.2$, $\gamma_I = \gamma_R = 0.4$ and the diffusion constant is fixed as $\sigma^2 = 10^{-2}$. The chosen compromise propensities imply that agents in the compartments $\{S,E\}$ change opinions through interactions more strongly than agents in the compartments $\{I,R\}$.

\begin{figure}
	\centering
	\subfigure[$f_S(w,t)$]{
	\includegraphics[scale = 0.3]{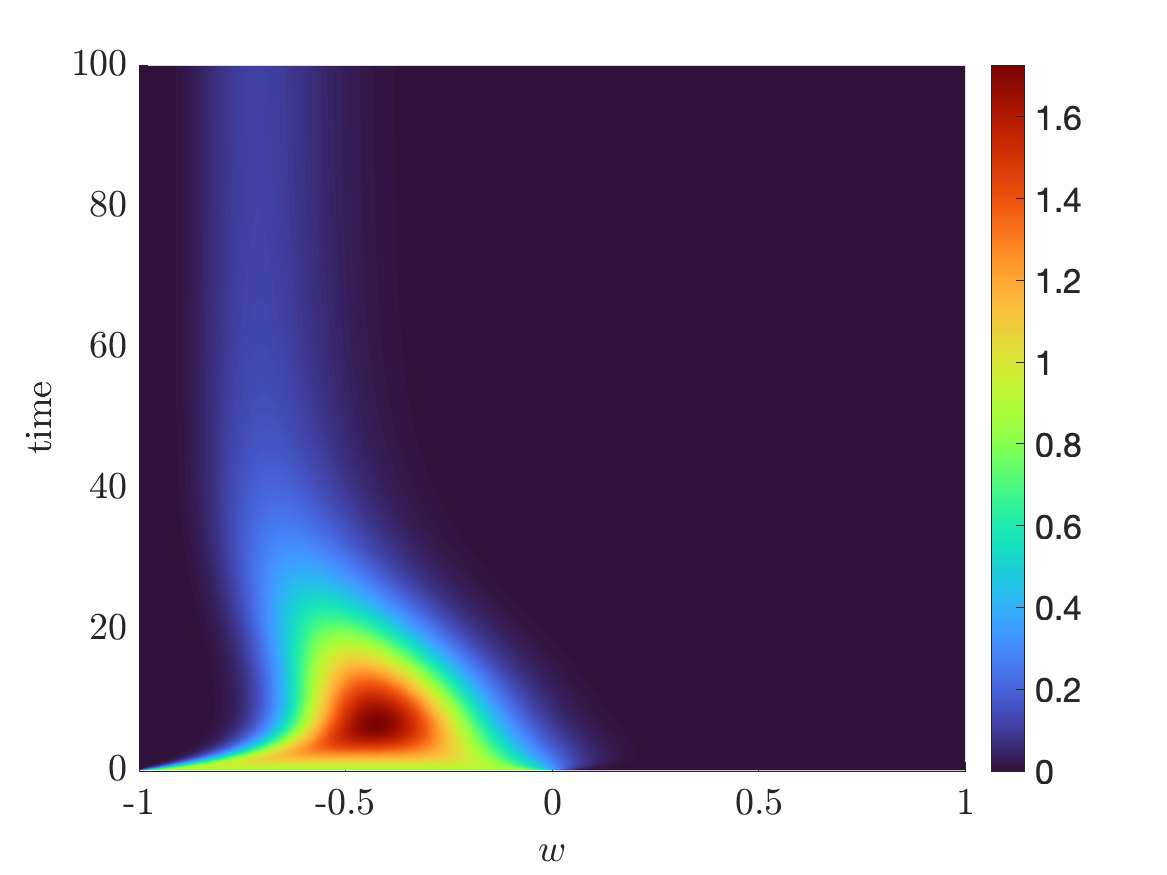}}
	\subfigure[$f_E(w,t)$]{
	\includegraphics[scale = 0.3]{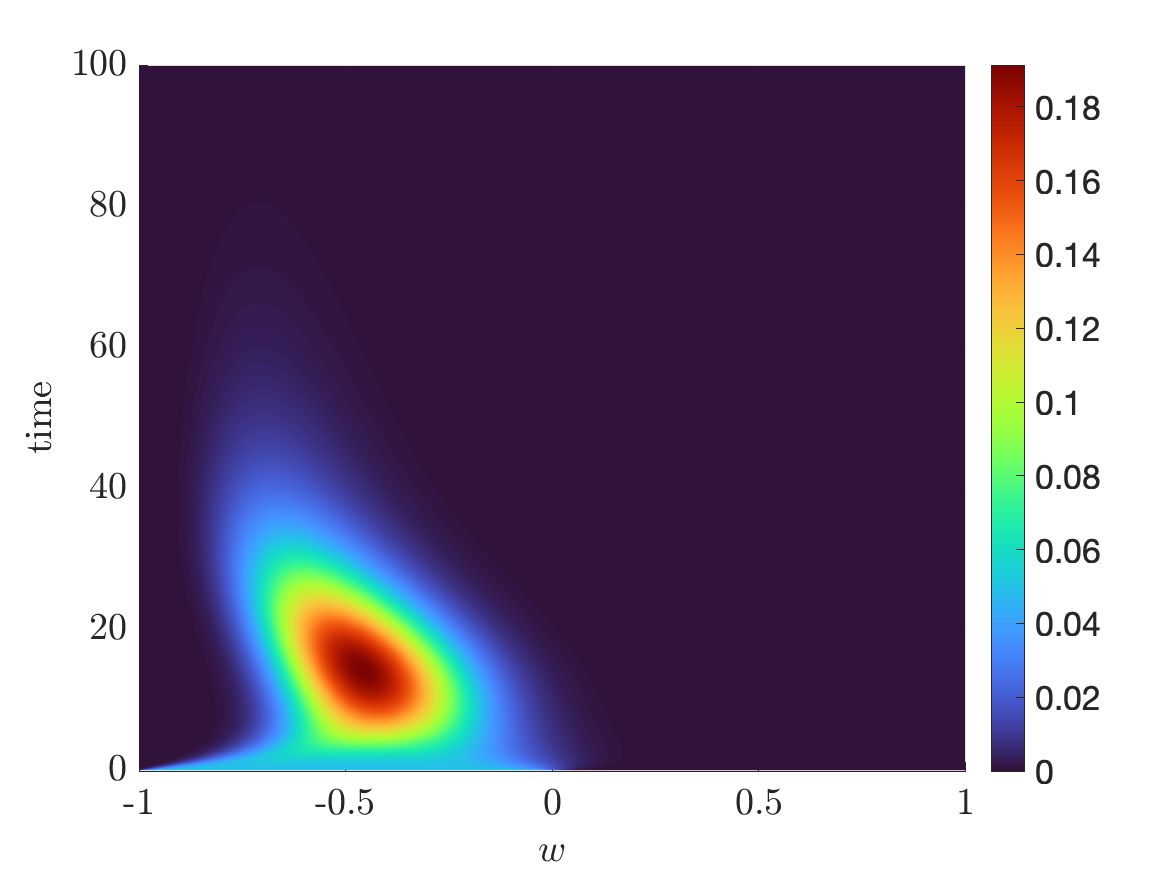}}	\\
	\subfigure[$f_I(w,t)$]{
	\includegraphics[scale = 0.3]{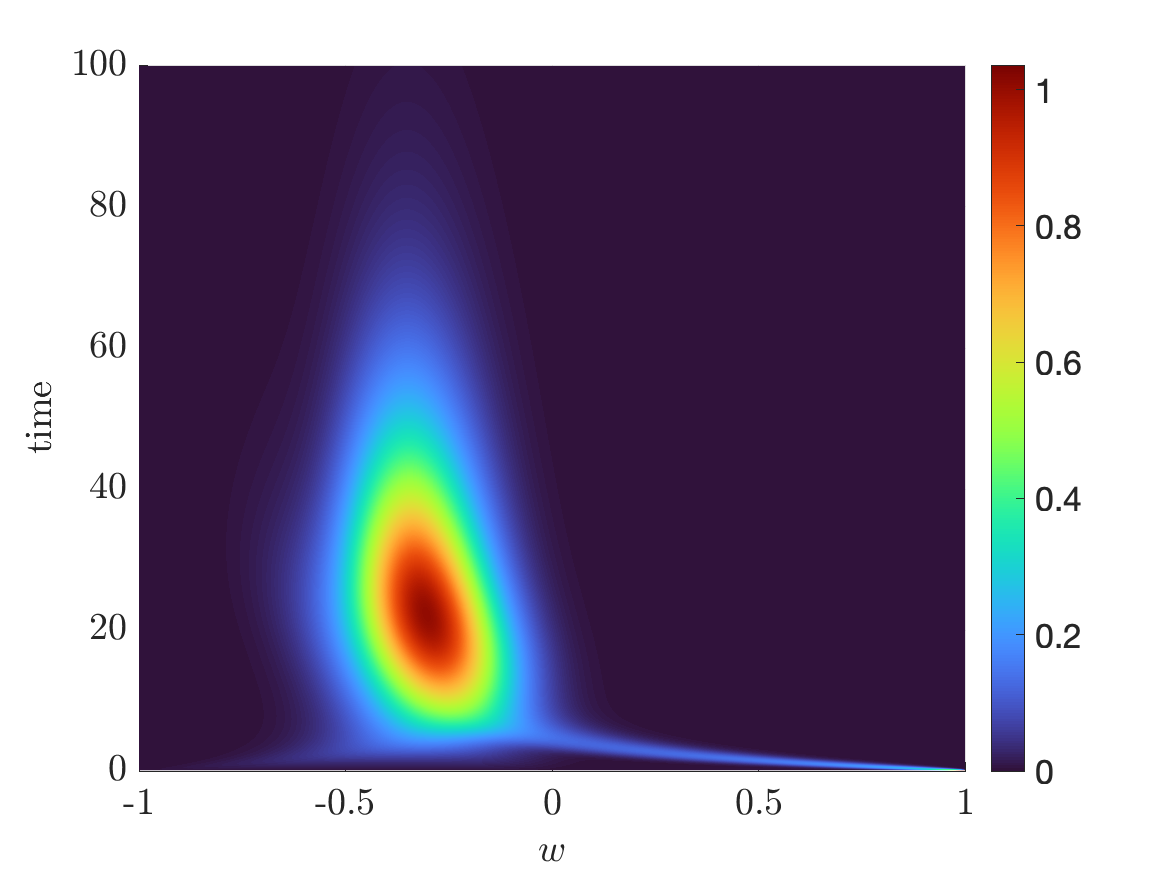}}
	\subfigure[$f_R(w,t)$]{
	\includegraphics[scale = 0.3]{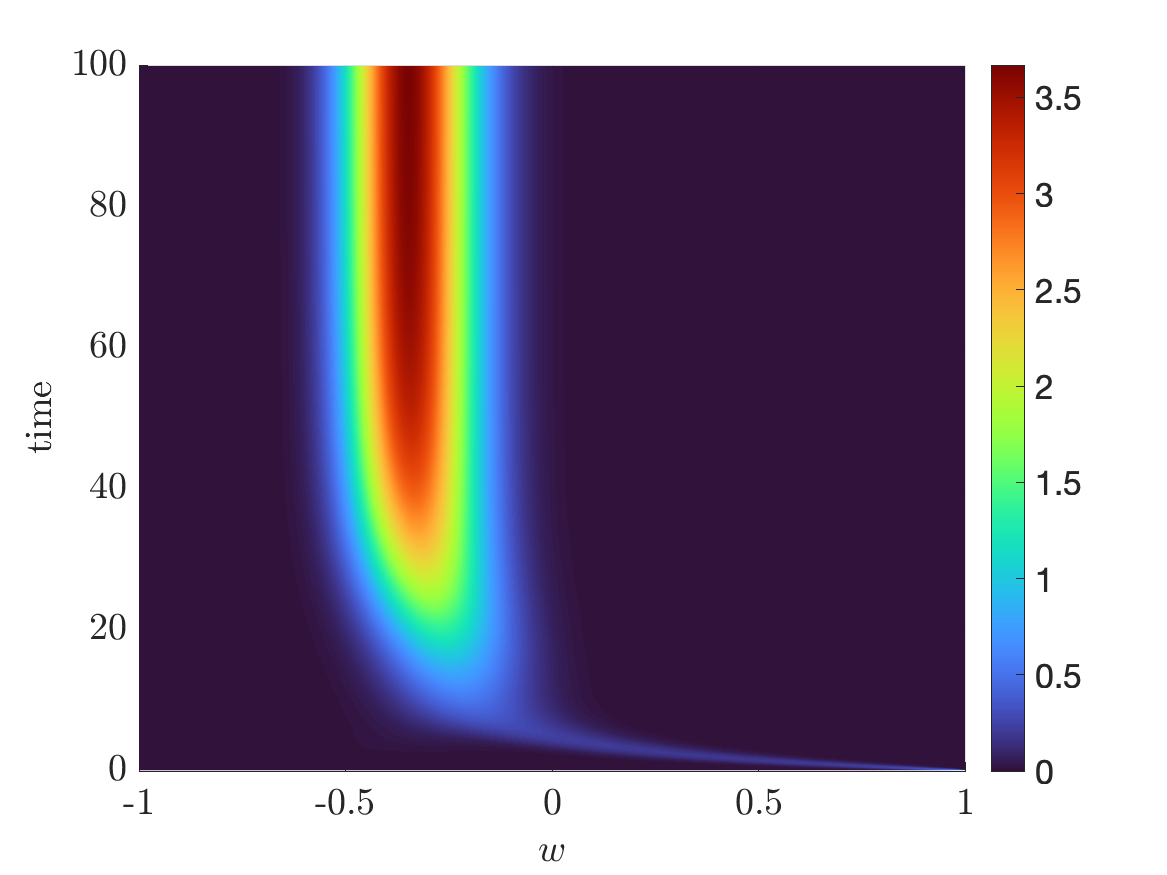}}
	\caption{\textbf{Test 1}. Evolution of the kinetic densities $f_J(w,t)$, $J \in \mathcal C$, over the time interval $[0,100]$, $\Delta t = 10^{-1}$. We considered the epidemic parameters $\beta = 0.3$, $\nu_E=1/2$, $\nu_I = 1/12$,  compromise propensities $\gamma_S = \gamma_E = 0.2$, $\gamma_I = \gamma_R = 0.4$, and diffusion constant $\sigma^2 = 10^{-2}$, the scaling parameter is $\tau = 1$. The discretization of the interval $[-1,1]$ is performed with $N_w = 501$ gridpoints. We fixed the initial condition as in \eqref{eq:ic_test1a} with $\rho_E(0)=\rho_I(0) =\rho_R(0) = 0.05$ and $\rho_S(0) = 1-\rho_E(0)-\rho_I(0)-\rho_R(0)$.  }
	\label{fig:1}
\end{figure}

We consider also the initial distributions 
\begin{equation}
\label{eq:ic_test1b}
\begin{split}
f_S(w,0) = \rho_S(0) h_0(w), \qquad f_E(w,0) = \rho_E(0)h_0(w), \\
f_I(w,0) = \rho_I(0) h_0(w), \qquad f_R(w,0) = \rho_R(0)h_0(w),
\end{split}
\end{equation}
with $\rho_E(0) = \rho_I(0) = \rho_R(0) = 0.05$ and $\rho_S(0) = 1-\rho_E(0)-\rho_I(0)-\rho_R(0)$. The defined initial conditions describe a society where all the agents share positive opinions towards the adoption of protective behaviour. We consider the same epidemiological parameters of the previous test and the same compromise propensities  and diffusion constant.
In Figure \ref{fig:2} we compare the evolution of the computed observable quantities obtained as $\int_{-1}^1 w^r f_J(w,t)dw$ with $\rho_J(t)$, $\rho_Jm_J(t)$ defined in the moment system \eqref{eq:SEIR}-\eqref{eq:mean_op} with the two sets of initial conditions. We may observe good agreement between the approximated evolution of observable quantities and the moment system. At the epidemiological level we may observe that, due to the hypothesis $\alpha=0$ which neglects opinion effects in transition between compartmens, the evolution of mass fractions $\rho_J(t)$ do not change in view of the two considered initial conditions. Anyway, thanks to the proposed kinetic approach we may obtain details on the evolution of mean opinions in each compartment. 
%

\begin{figure}
	\centering
	\includegraphics[scale = 0.3]{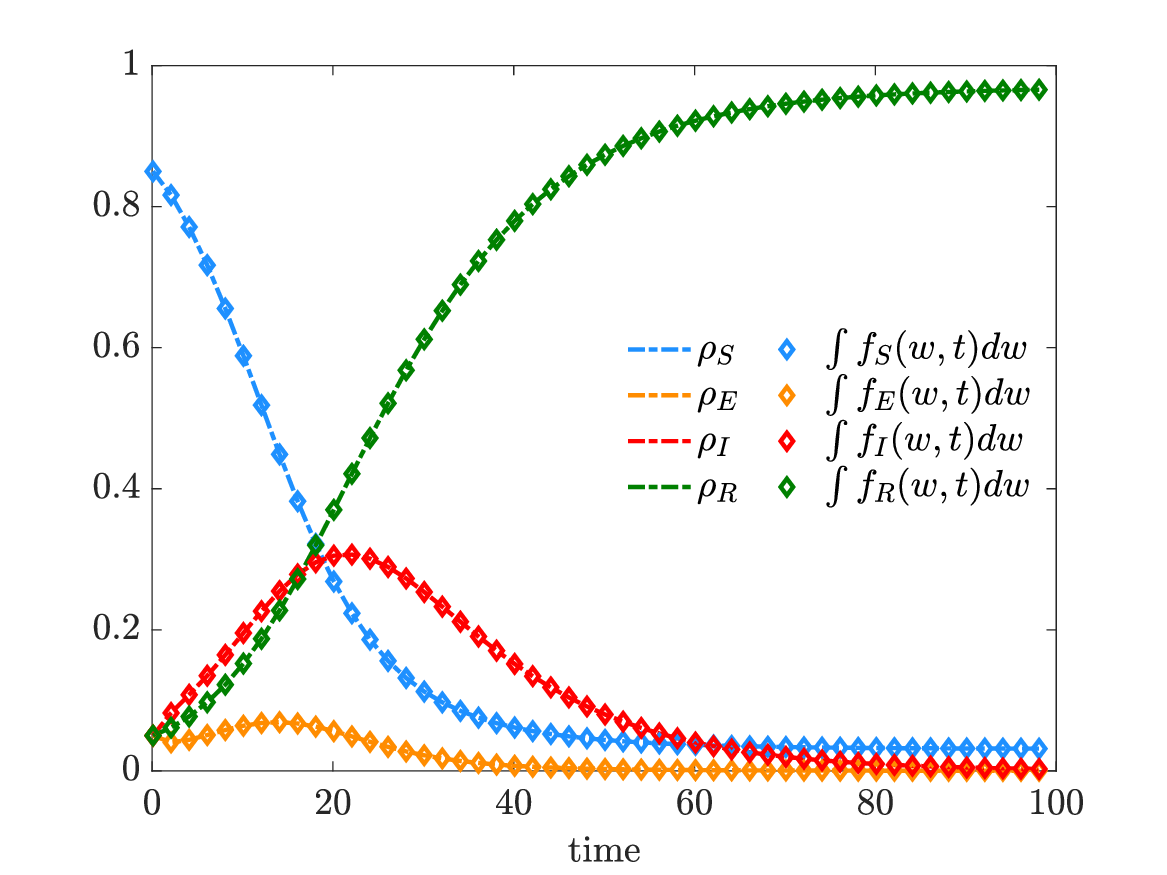}
	\includegraphics[scale = 0.3]{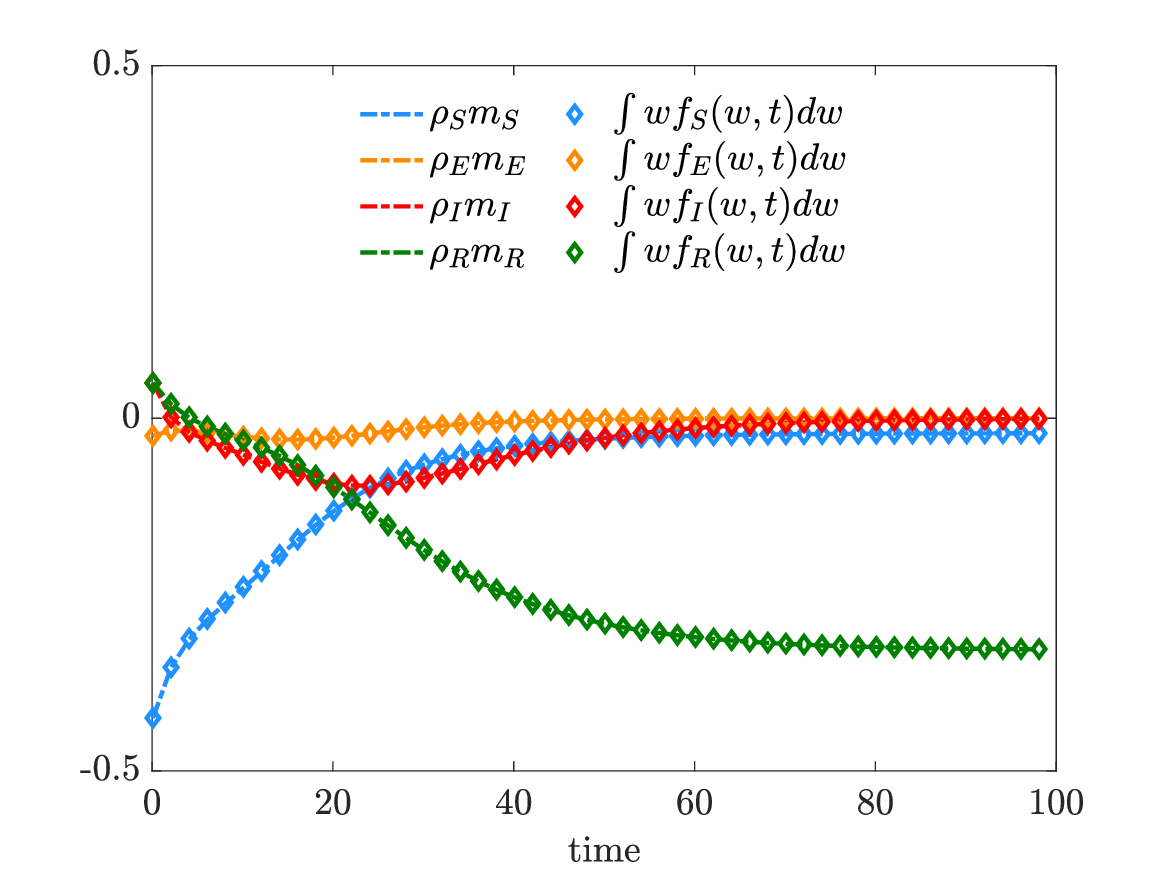} \\
	\includegraphics[scale = 0.3]{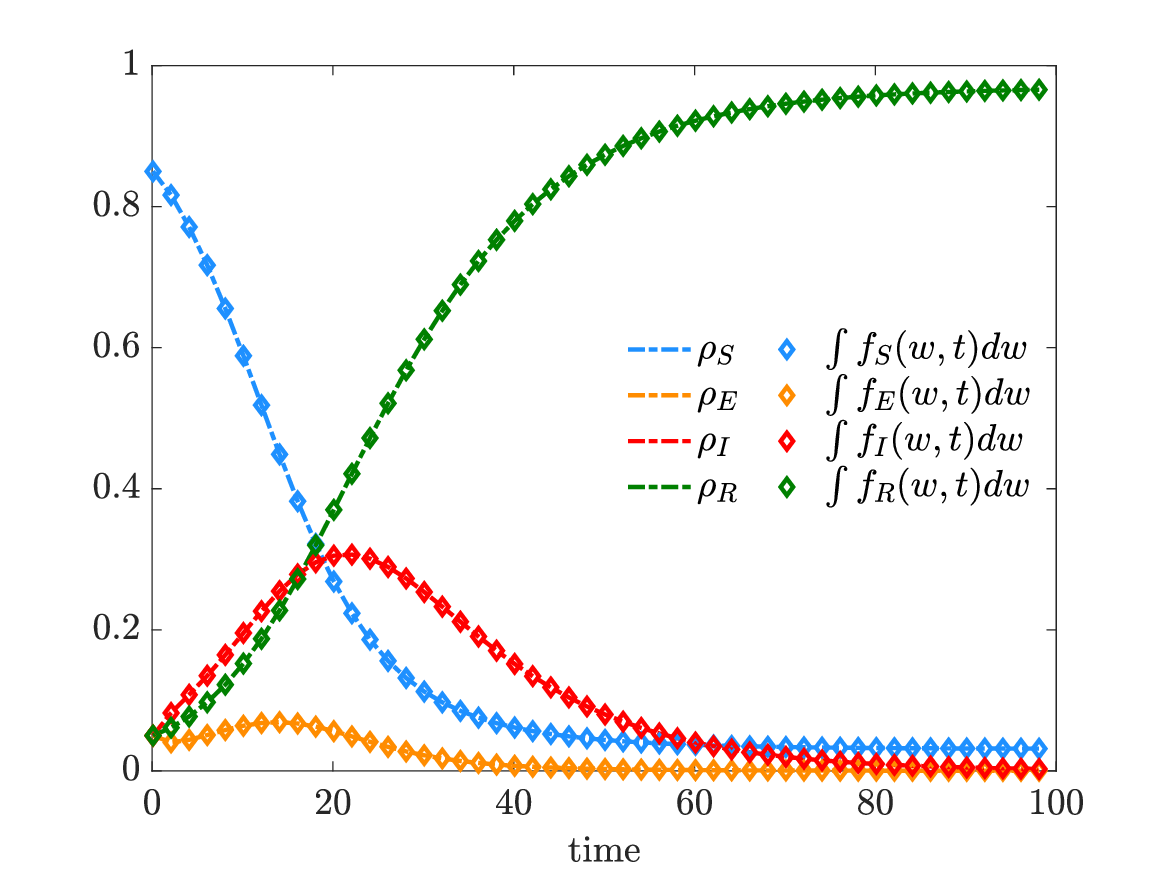}
	\includegraphics[scale = 0.3]{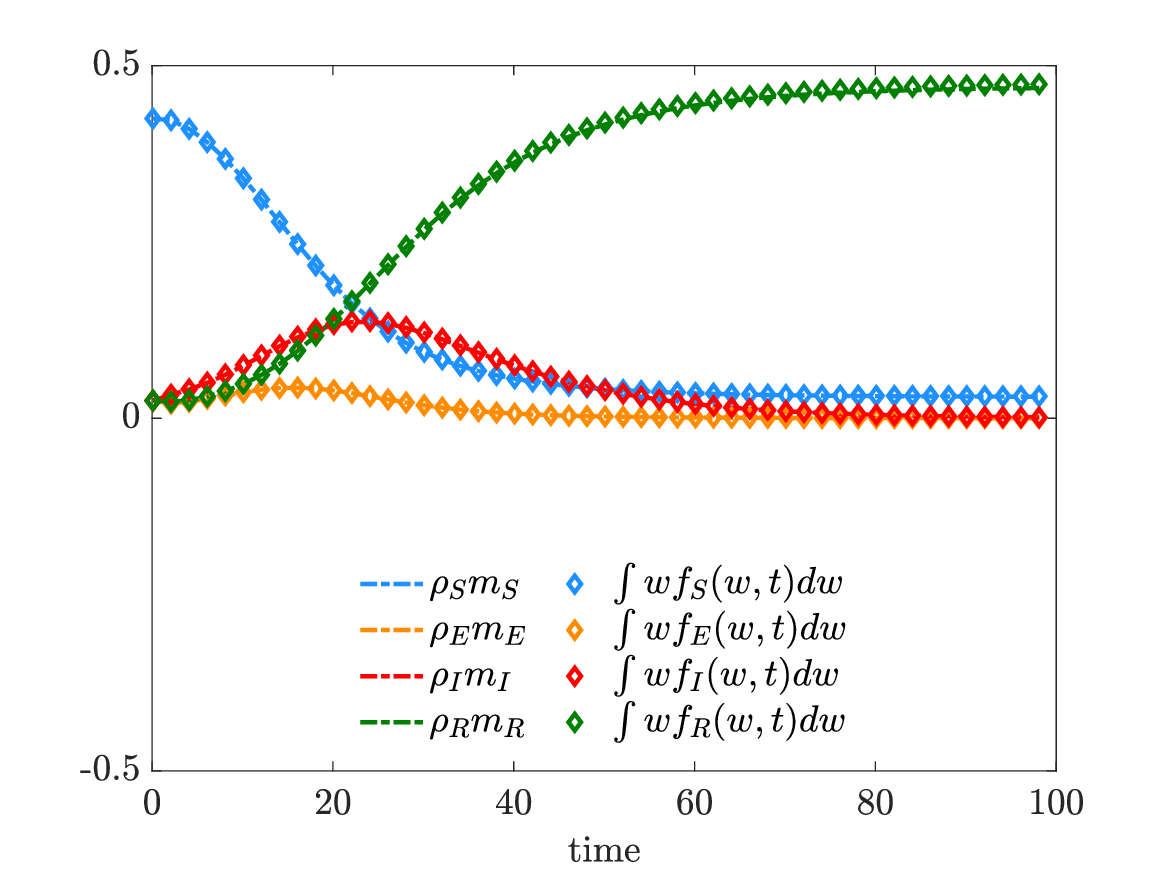}
	\caption{\textbf{Test 1}. Comparison between the evolution of $\rho_J$, $m_J\rho_J$ solution to the \rev{moments} system \eqref{eq:SEIR}-\eqref{eq:mean_op} and mass and momentum obtained from the numerical solution to \eqref{eq:kin_cs,alpha=0}. Top row: initial condition defined in \eqref{eq:ic_test1a}. Bottom row: initial condition defined in \eqref{eq:ic_test1b}. The epidemiological and numerical parameters have been fixed as in Figure \ref{fig:1}. }
	\label{fig:2}
\end{figure}

\subsection{Test 2: Opinion-dependent incidence rate}
In this test we investigate the influence of the initial conditions in a kinetic compartmental model with opinion-dependent local incidence rate of the form \eqref{def:K,alpha=1}. In particular, consider $\kappa(w,w_*)$ in \eqref{def:kappa} with $\alpha=1$ and we integrate the kinetic model \eqref{eq:kin_cs,alpha=1} on the time frame $[0,100]$, $\Delta t = 10^{-1}$ by considering a positively skewed population, synthesized in the following initial condition
\[
\begin{split}
\textrm{(IC1):}\quad&f_S(w,0) = \rho_S(0)h_1(w),\qquad f_E(w,0)=\rho_E(0)h_1(w), \\
&f_I(w,0) = \rho_I(0)h_1(w),\qquad f_R(w,0)=\rho_R(0)h_1(w), 
\end{split}
\]
with a negatively skewed population, obtained by considering the following initial condition
\[
\begin{split}
\textrm{(IC2):}\quad&f_S(w,0) = \rho_S(0)g_1(w),\qquad f_E(w,0)=\rho_E(0)g_1(w), \\
&f_I(w,0) = \rho_I(0)h_1(w),\qquad f_R(w,0)=\rho_R(0)h_1(w). 
\end{split}
\]
where 
\[
g_1(w)
\begin{cases}
2 & w \in [-\frac{1}{2},-1] \\
0 & \textrm{elsewhere}
\end{cases}, \qquad
h_1(w)
\begin{cases}
2 & w \in [\frac{1}{2},1]\\
0 &  \textrm{elsewhere}.
\end{cases}
\]
In both cases we fixed $\rho_E(0) = \rho_I(0) = \rho_R(0) = 0.05$ and $\rho_S(0) = 1-\rho_E(0)-\rho_I(0)-\rho_R(0)$. In Figure \ref{fig:compar_alpha1} we depict the evolution of kinetic mass and momentum obtained from \eqref{eq:kin_cs,alpha=1} with respectively initial conditions $\textrm{(IC1)}$ or $\textrm{(IC2)}$. We can observe that, at variance with what we obtained in Section \ref{test1}, an opinion-dependent incidence rate effectively quantifies the impact of opinion-type dynamics on the epidemic evolution. Indeed, in the case $\textrm{IC1}$, where the agents' opinions tends to align towards protective behaviours, the transmission dynamics become sensitive to the introduced social dynamics. 

\begin{figure}
	\centering
	\includegraphics[scale = 0.3]{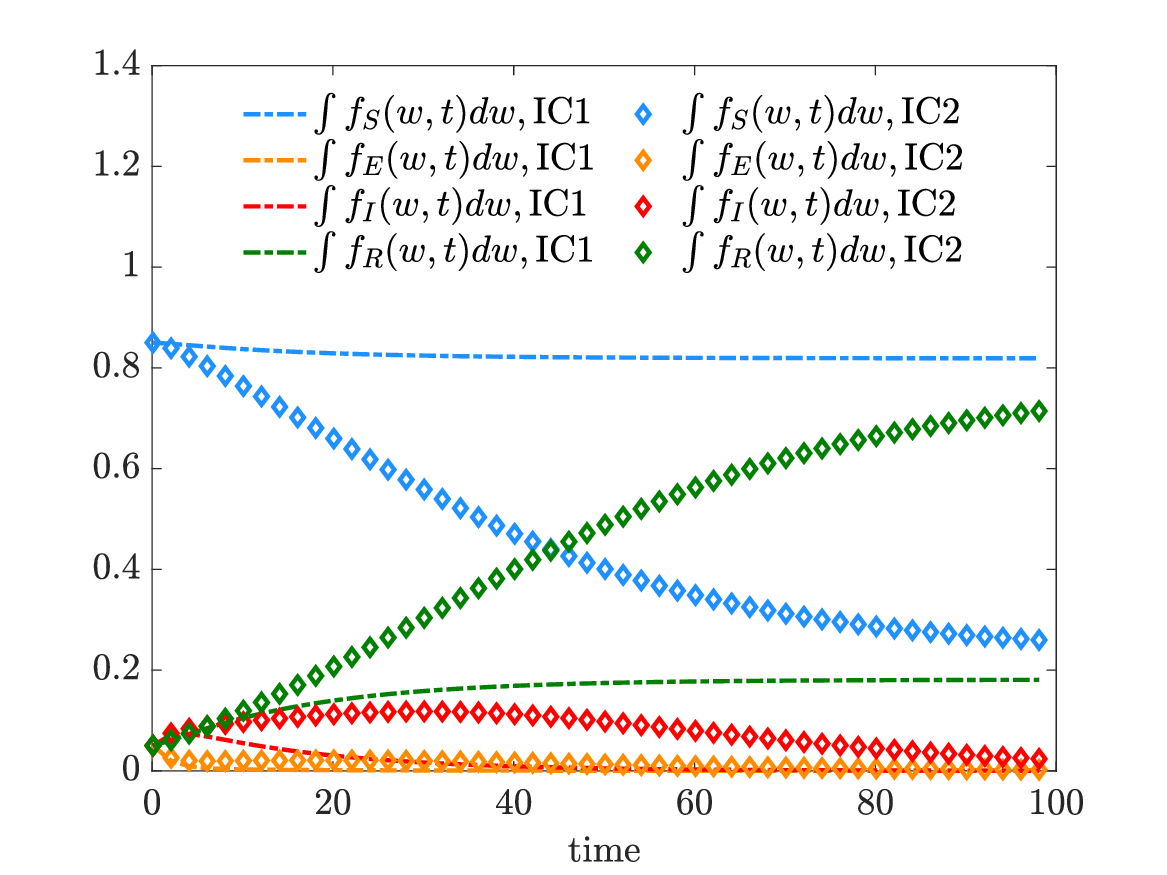}
	\includegraphics[scale = 0.3]{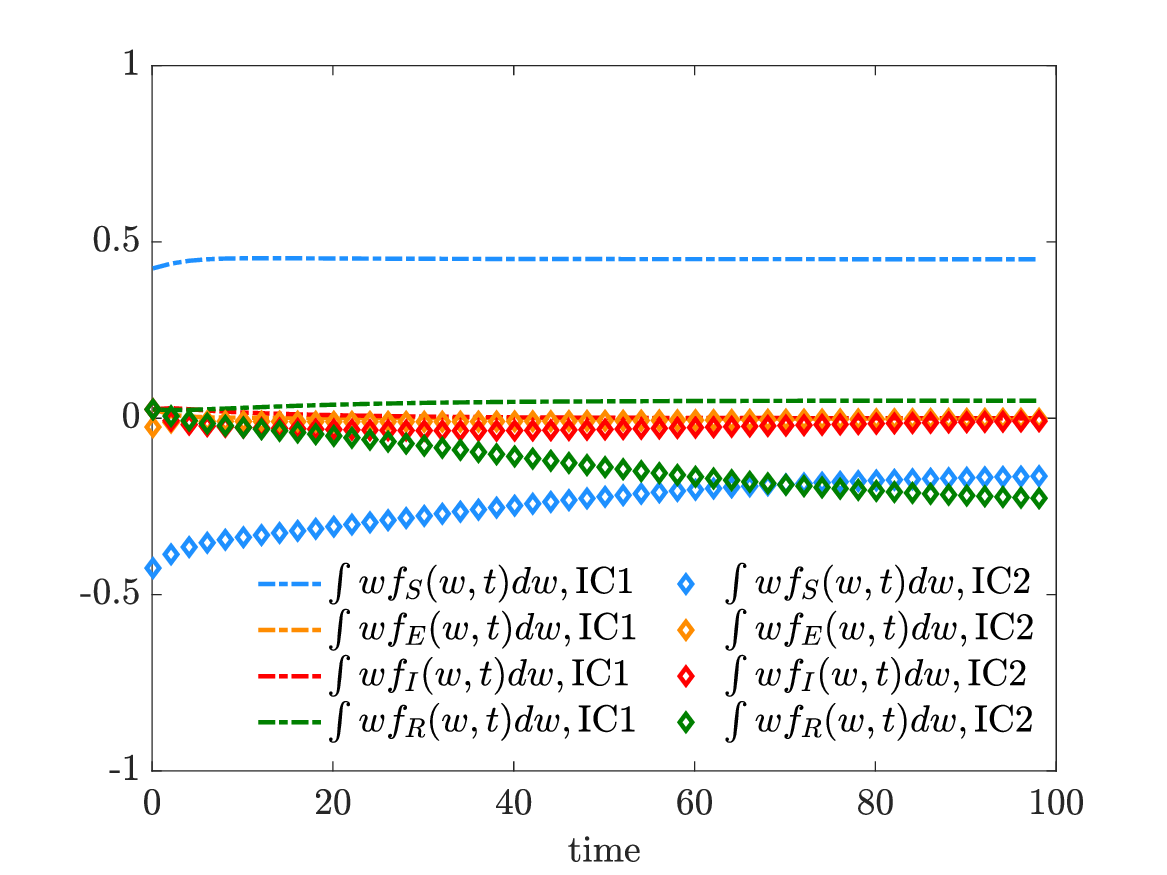} \\
	\caption{\textbf{Test 2}. Comparison between the evolution of $\rho_J$, $m_J\rho_J$ solution to the moment system \eqref{eq:mass,alpha=1}-\eqref{eq:mean_op,alpha=1} and mass and momentum obtained from the numerical solution to kinetic system \eqref{eq:kin_cs,alpha=1}.
		The epidemiological and numerical parameters have been fixed as in Figure \ref{fig:1} and the initial conditions as in $\textrm{(IC1)}$-$\textrm{(IC2)}$. }
	\label{fig:compar_alpha1}
\end{figure}

\subsection{Test 3. Impact of opinion clusters on the epidemic dynamics}

In this test we focus on the effects of the asymptotic formation of opinion clusters as discussed in Section \ref{subs:stationary_sol}. We consider the epidemiological parameters defined in the previous tests, $\beta = 0.3$, $\nu_I = 1/12$, $\nu_E = 1/2$. Furthermore we fix as initial conditions the one defined in \eqref{eq:ic_test1a} with $\rho_E(0) =\rho_I(0) = \rho_R(0) = 0.05$ and $\rho_S(0) = 1-\rho_E(0)-\rho_I(0)-\rho_R(0)$. The opinion formation dynamics is solved through a semi-implicit SP scheme over an uniform grid for $[-1,1]$ composed by $N_w = 501$ gridpoints and a time discretization of the time horizon $[0,100]$ obtained with $\Delta t = 10^{-1}$. The parameters characterizing the opinion dynamics are $\gamma_S = \gamma_E = 0.8$, $\gamma_I = \gamma_R = 0.2$ such that the susceptible and  the exposed populations, which are initially skewed towards negative opinions, weights more opinions of other compartments as in \eqref{eq:binary_rules}. Furthermore we consider a diffusion $\sigma^2 = 10^{-3}$. We remark that these choices are coherent with the ones adopted to obtain Figure \ref{fig:bimodal_cases} $(b)$.

In Figure \ref{fig:test3} we present the evolution of the total density $f(w,t) = \sum_{J \in \mathcal C} f_J(w,t)$ for several choices of the parameter $\alpha\ge0$ in the local incidence rate expressed by \eqref{def:kappa}. In the regime $\alpha=0$, as highlighted in Section \ref{subs:stationary_sol}, we detect the formation of clusters. We may observe how opinion clusters appear also in regimes $\alpha>0$ and may lead to stationary profiles of different nature with respect to the one obtained with $\alpha=0$. The emergence of opinion clusters can be therefore obtained in more general regimes where the transmission dynamics depends on the behaviour of infected agents. 

\begin{figure}
\centering
\includegraphics[scale = 0.225]{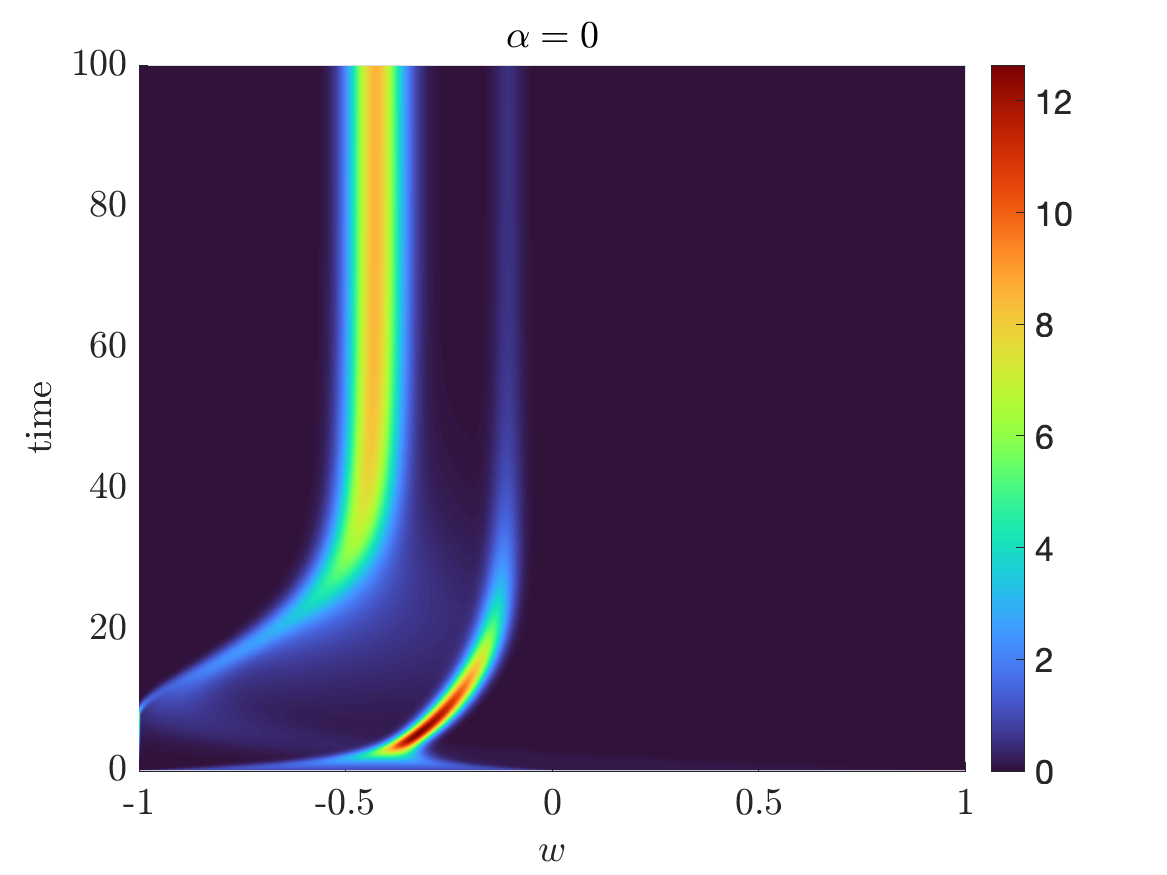}
\includegraphics[scale = 0.225]{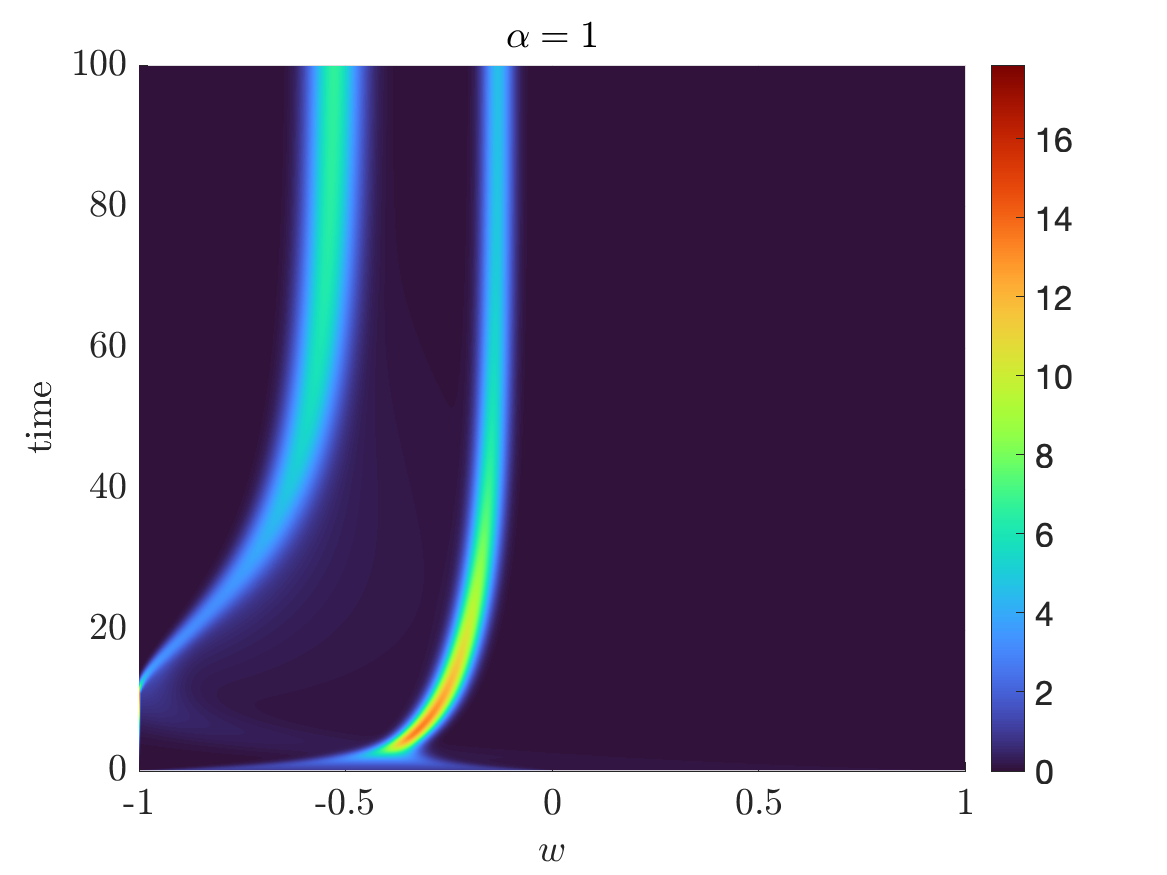}
\includegraphics[scale = 0.225]{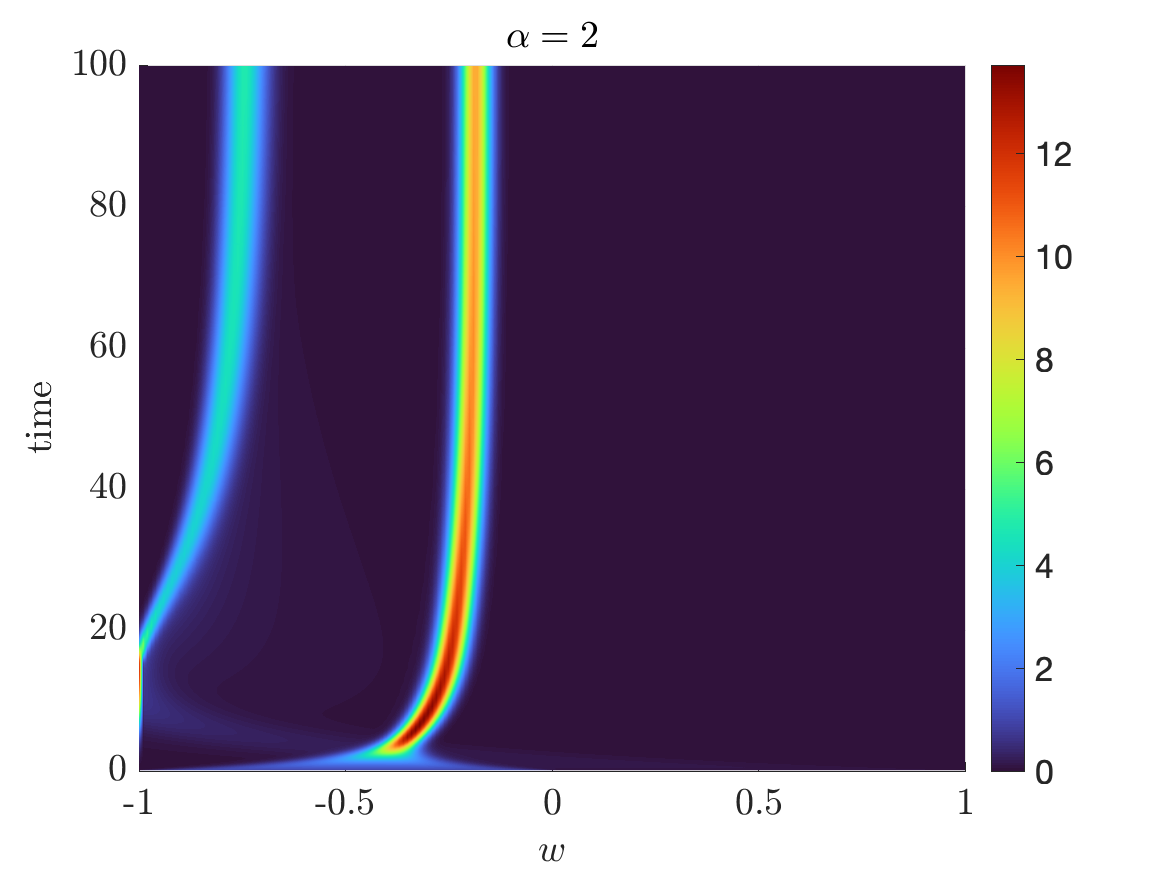}\\
\caption{\textbf{Test 3}. Evolution of the density $f(w,t) = \sum_{J\in \mathcal C}f_J(w,t)$ over the time horizon $[0,100]$, where $f_J(w,t)$ are the numerical solutions to\eqref{ivp:FP_system} with $\beta = 0.3$, $\nu_I = 1/12$, $\nu_E = 1/2$ and $\gamma_S = \gamma_E = 0.8$, $\gamma_I = \gamma_R = 0.2$. We considered $\alpha = 0$ (left), $\alpha = 1$ (center), $\alpha = 2$ (right). Initial condition as in \eqref{eq:ic_test1a}.   }
\label{fig:test3}
\end{figure}

As discussed in the case of explicitly solvable stationary solution, the value of the diffusion $\sigma^2>0$ is of great importance to determine the emergence of opinion clusters and of polarization. The impact of the steady state on the epidemic dynamics is studied in Figure \ref{fig:rhoR_test3} where we integrate \eqref{ivp:FP_system} over the time integral $[0,T]$, $T = 200$, for several values of the diffusion constant $\sigma^2 \in [10^{-4}, 0.2]$ and we consider the large-time mass of recovered individuals $\rho_R(T) = \int_{\mathcal I} f_R(w,T)dw$ for several values of $\alpha = 0,1,2$. As before, we fixed the compromise parameters $\gamma_S = \gamma_E = 0.8$, $\gamma_I = \gamma_R = 0.2$. We may observe how, under the aforementioned conditions, large values of the diffusion parameters trigger a higher number of recovered individuals. This is due to the emergence of polarization in the society which is driven towards negative opinions under the considered initial condition. 

\begin{figure}
\centering
\includegraphics[scale = 0.225]{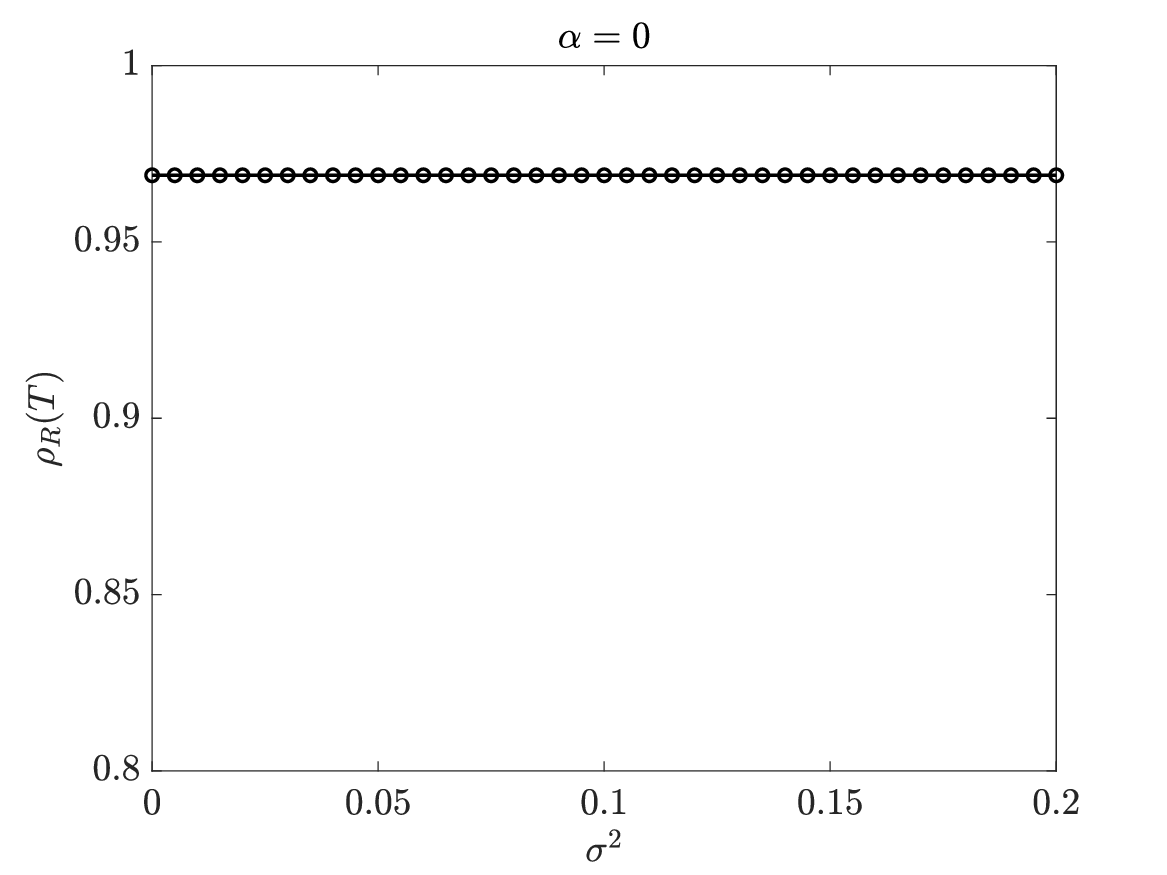}
\includegraphics[scale = 0.225]{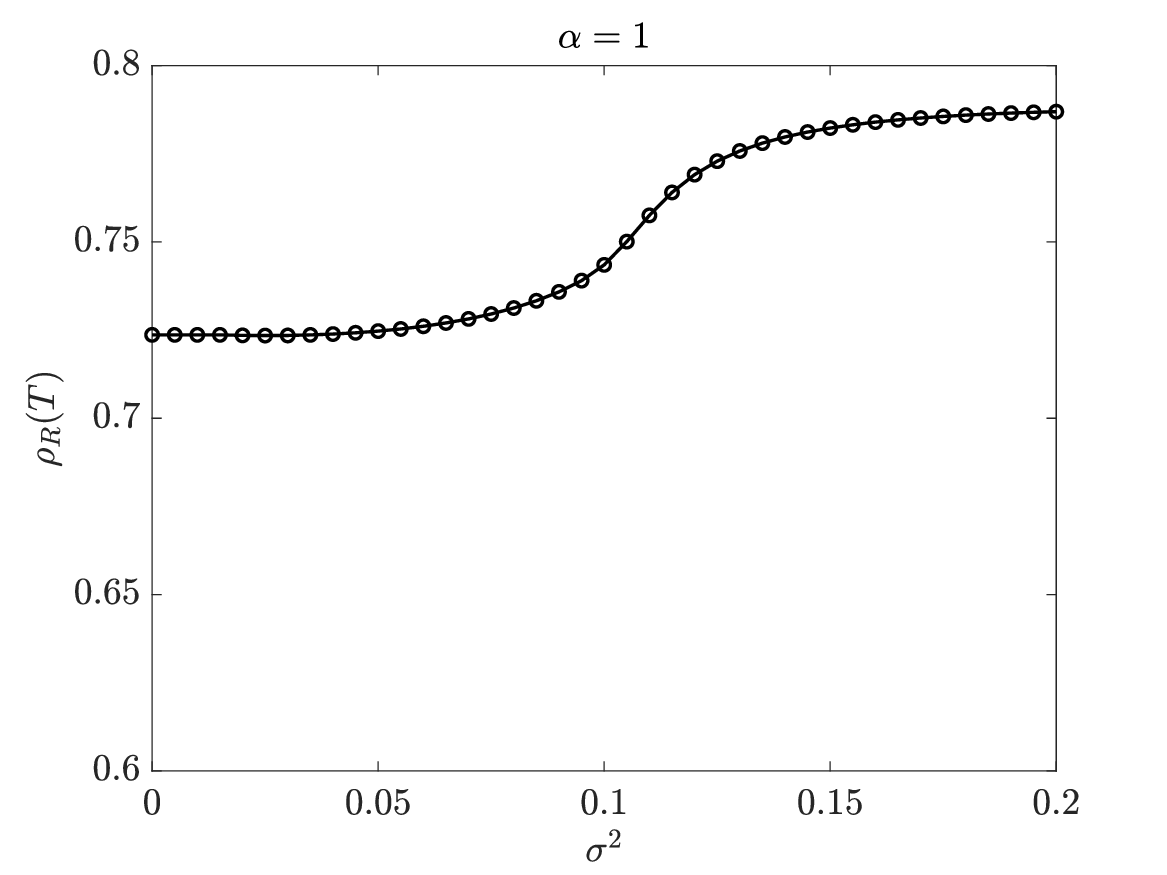}
\includegraphics[scale = 0.225]{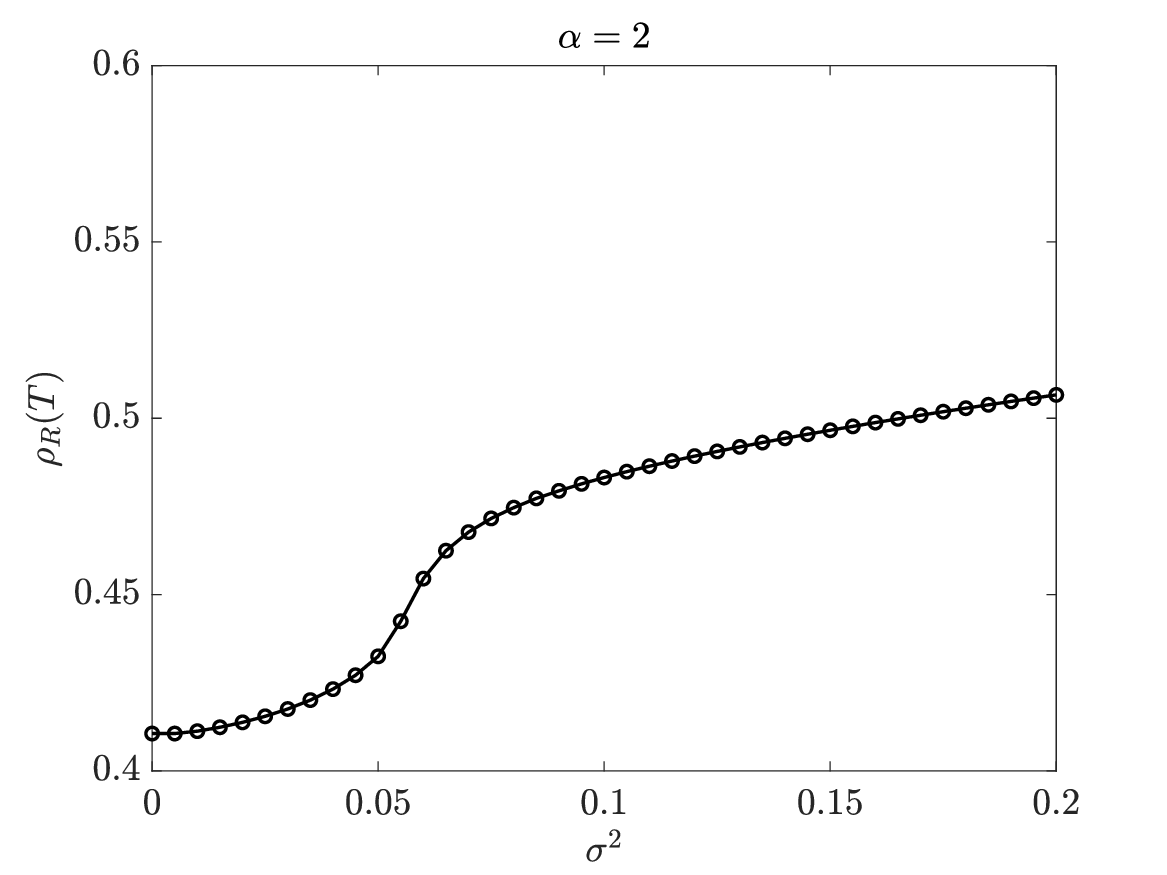}
\caption{\textbf{Test 3}: we depict $\rho_R(T)  = \int_{\mathcal I} f_R(w,T)dw$ with $T = 200$ obtained with numerical integration of \eqref{ivp:FP_system} with initial condition \eqref{eq:ic_test1a}, $\beta = 0.3$, $\nu_I = 1/12$, $\nu_E = 1/2$. Numerical integration performed over $[0,T]$, $T = 200$ with $\Delta t = 10^{-1}$ and a discretization of $\mathcal I$ obtained with $N_w = 501$ gridpoints.   }
\label{fig:rhoR_test3}
\end{figure}

\section*{Conclusion}

In this work we \rev{focused} on the development of a kinetic model for the interplay between opinion and epidemic dynamics. The study of the impact of opinion-type phenomena in the evolution of infectious diseases \rev{can} be suitably linked with vaccine hesitancy. Recently this phenomenon emerged in close connection with the evolution of pandemics. In this paper, we \rev{studied} the evolution of opinion densities by means of a compartmental kinetic model where the microscopic interaction dynamics is supposed heterogeneous with respect to the agents' \rev{compartments}. Through explicit computations, we showed the formation of asymptotic clusters for a surrogate Fokker-Planck-type model under the assumption that the transmission dynamics is independent \rev{of} opinion-formation processes. Furthermore, we studied positivity and uniqueness of the solution of the model. Numerical experiments confirm the ability of the approach to force clusters formation also in the case of opinion-dependent  transmission dynamics.  Future studies will aim \rev{at defining} the parameters of the model by resorting to existing experimental data.

\section*{Acknowledgments}
This work has been written within the activities of GNFM group of INdAM (National Institute of High Mathematics). MZ acknowledges the support of MUR-PRIN2020 Project No.2020JLWP23 (Integrated Mathematical Approaches to Socio-Epidemiological Dynamics). This work has been supported by NextGenerationEU, MZ wishes to acknowledge the National Centre for HPC, Big Data and Quantum Computing (CN00000013).
The work of SB is funded by the Deutsche Forschungsgemeinschaft (DFG, German Research Foundation) – 320021702/GRK2326 – Energy, Entropy, and Dissipative Dynamics (EDDy).



\end{document}